\documentclass[12pt, onecolumn]{IEEEtran}

\usepackage{cite}
\usepackage{pslatex,bm}
\usepackage{epsfig}
\usepackage{amsmath}
\usepackage{amssymb}
\usepackage{amsthm}
\usepackage{booktabs}
\usepackage{array}
\usepackage{multicol}
\usepackage{color}
\usepackage{url}
\usepackage{dsfont}
\usepackage{tikz}
\usepackage{cases}
\usepackage{threeparttable}
\usetikzlibrary{arrows,shapes.geometric,backgrounds,calc}
\psfull
\usepackage{makeidx}
\usepackage{lastpage}
\usepackage[toc]{glossaries}
\newtheorem{Theorem}{Theorem}

\newtheorem{Ex}{Example}
\newtheorem{Corollary}{Corollary}
\newtheorem{definition}{Definition}

\newtheorem{rmk}{\sf Remark}

\newtheorem{Prop}{Proposition}

\newcommand*{\QEDA}{\hfill\ensuremath{\blacksquare}}%

\DeclareMathAlphabet{\mathpzc}{OT1}{pzc}{m}{it}
\renewcommand{\baselinestretch}{1.9}
\newacronym{CSI}{CSI}{Channel State In\-for\-ma\-tion}

\begin{document}
{{\renewcommand{\baselinestretch}{1.2}
\author{\authorblockN{Pin-Hsun Lin, \IEEEmembership{Member, IEEE}, Eduard A. Jorswieck, \IEEEmembership{Senior Member, IEEE},\\ Rafael F. Schaefer, \IEEEmembership{Senior Member, IEEE}, Martin Mittelbach, and \\Carsten R. Janda, \IEEEmembership{Student Member, IEEE}}
\thanks{Parts of the work were presented in ITW 2016 \cite{PH_ITW16}, Cambridge, UK and SCC 2017 \cite{PH_SCC17}, Hamburg, Germany. Pin-Hsun Lin, Eduard A. Jorswieck, Martin Mittelbach and Carsten R. Janda are with the Communications Laboratory, Department of Electrical Engineering and Information Technology, Technische Universit\"{a}t Dresden, Germany. Rafael F. Schaefer is with the Information Theory and Applications Chair, Technische Universit\"{a}t Berlin, Germany.
Emails: {\{pin-hsun.lin, eduard.jorswieck, martin.mittelbach, carsten.janda\}@tu-dresden.de, rafael.schaefer@tu-berlin.de}. Part of this work is funded by FastCloud 03ZZ0517A.} }
}
\title{On Stochastic Orders and Fading Multiuser Channels with Statistical CSIT}
\maketitle \thispagestyle{empty}

{\renewcommand{\baselinestretch}{1.5}
\vspace{-2cm}
\begin{abstract}
%
In this paper, fading Gaussian multiuser channels are considered. If the channel is perfectly known to the transmitter, capacity has been established for many cases in which the channels may satisfy certain information theoretic orders such as degradedness or strong/very strong interference.
Here, we study the case when only the statistics of the channels are known at the transmitter which is an open problem in general. The main contribution of this paper is the following: First, we introduce a framework to classify random fading channels based on their joint distributions by leveraging three schemes: maximal coupling, coupling, and copulas.
The underlying spirit of all scheme is, we obtain an equivalent channel by changing the joint distribution in such a way that it now satisfies a certain information theoretic order while ensuring that the marginal distributions of the channels to the different users are not changed. The construction of this equivalent multi-user channel allows us to directly make use of existing capacity results, which includes Gaussian interference channels, Gaussian broadcast channels, and Gaussian wiretap channels.
We also extend the framework to channels with a specific memory structure, namely, channels with finite-state, wherein the Markov fading broadcast channel is discussed as a special case. Several practical examples such as Rayleigh fading and Nakagami-\textit{m} fading illustrate the applicability of the derived results.
\end{abstract}
\vspace{-0.5cm}
{\renewcommand{\baselinestretch}{1.5}
%

%

\section{Introduction}\label{Sec_Intro}

For some Gaussian multiuser (GMU) channels with perfect channel state information at the transmitter (CSIT), due to the capability of ordering the marginal channels of different users, capacity regions have been successfully derived. These include the degraded broadcast channel (BC) \cite{Bergmans_degraded_BC, Gallager_degraded_BC}, the wiretap channel (WTC) \cite{Khisti_MIMOME,Oggier_MIMOME}, and the interference channels (IC) with \textit{strong} interference \cite{Sato_IC,Carleial_IC,HK_IC}, with \textit{very strong} interference \cite{Carleial_IC}, and in the low-interference regime
\cite{Veeravalli_sum_capacity_IC}. When fading effects of wireless channels are taken into account, capacity results of some channels have been found for perfect CSIT. For example, in \cite{Liang_fading_secrecy}, the ergodic secrecy capacity of Gaussian WTC (GWTC) is derived; in \cite{Sankar_ergodic_GIC}, the ergodic capacity regions are derived for \textit{ergodic very strong} and \textit{uniformly strong} Gaussian IC (GIC), where each realization of the fading process is a strong IC.
In practice however, due to limited feedback bandwidth and the delay caused by channel estimation, the transmitter may not be able to track channel realizations perfectly and instantaneously making the assumption of perfect CSIT too ambitious. Thus, for fast fading channels, it is more realistic to consider the case with only partial or delayed CSIT. In particular, when there is solely statistical CSIT available, capacity is known only in very few cases such as the layered BC \cite{Tse_fading}, the binary fading interference channel \cite{Vahid_binary_fading_IC}, the one-sided layered IC \cite{Guo_ZIC}, GWTC \cite{SC_TIFS}, etc. Deriving the capacity of multiuser channels usually relies on information theoretic (IT) orders such as \textit{degraded}, \textit{less noisy}, and \textit{more capable} \cite{Korner_less_noisy, Kim_lecture}, etc., which allow to order the channels accordingly. In the lack of instantaneous CSIT, identifying whether an MU channel satisfies a certain IT order or not, is usually not obvious which makes this approach of deriving capacity results much more involved. Taking fading GBC as an example, only capacity bounds can be found in \cite{Tuninetti_DIMACS03}, \cite{Jafarian_ISIT11_fading_GBC}, and \cite{Farsani_ISIT13_fading_GBC}.

In the following we give a simple motivating example from a two-user fading GBC. Without loss of generality we assume that the means and variances of the additive white Gaussian noises (AWGN) at different receivers are identical. Denote the \textit{channel gain}\footnote{In this paper we use channel \textit{gain} to denote the square of the channel magnitude.} of the two real random channels by $H_1$ and $H_2$ with probability density functions (PDF) $f_{H_1}$ and $f_{H_2}$, respectively. In Fig. \ref{Fig_channel_distr}(a), $f_{H_1}$ and $f_{H_2}$ do not intersect. Therefore, even when there is only statistical CSIT, we still can know that the realizations of $H_1$ and $H_2$, namely, $h_1$ and $h_2$, respectively, always satisfy  $h_1<h_2$. Then channel 1 is degraded with respect to channel 2. In contrast, in Fig. \ref{Fig_channel_distr}(b), the intersection of the supports of the two channels is not empty. Therefore, the trichotomy order\footnote{In order to make a consistent presentation when compare to the stochastic order, in the following, we will use \textit{trichotomy order} instead of \textit{trichotomy law} to show the three relations between two deterministic scalar variables $a$ and $b$: $a>b$, $a=b$, and $a<b$.}\cite{Apostol_calculus} of the realizations $h_1$ and $h_2$ may alter over time within a codeword length. A sufficient condition for a memoryless channel to satisfy a certain IT order is that it must be satisfied over the whole codeword length. Therefore, the transmitter is not able to directly identify the degradedness between the channels in Fig. \ref{Fig_channel_distr}(b) by just comparing $h_1$ and $h_2$. Based on the above observation, in this paper we address the following unsettled questions for GMU channels with only statistical CSIT:
\begin{enumerate}
\item How to efficiently compare channel gains solely based on their distributions with the goal to verify whether they satisfy a certain IT order or not?
\item How to derive the capacity region by exploiting such a comparison of channel gains?
\end{enumerate}

\begin{figure}[h]
\centering \epsfig{file=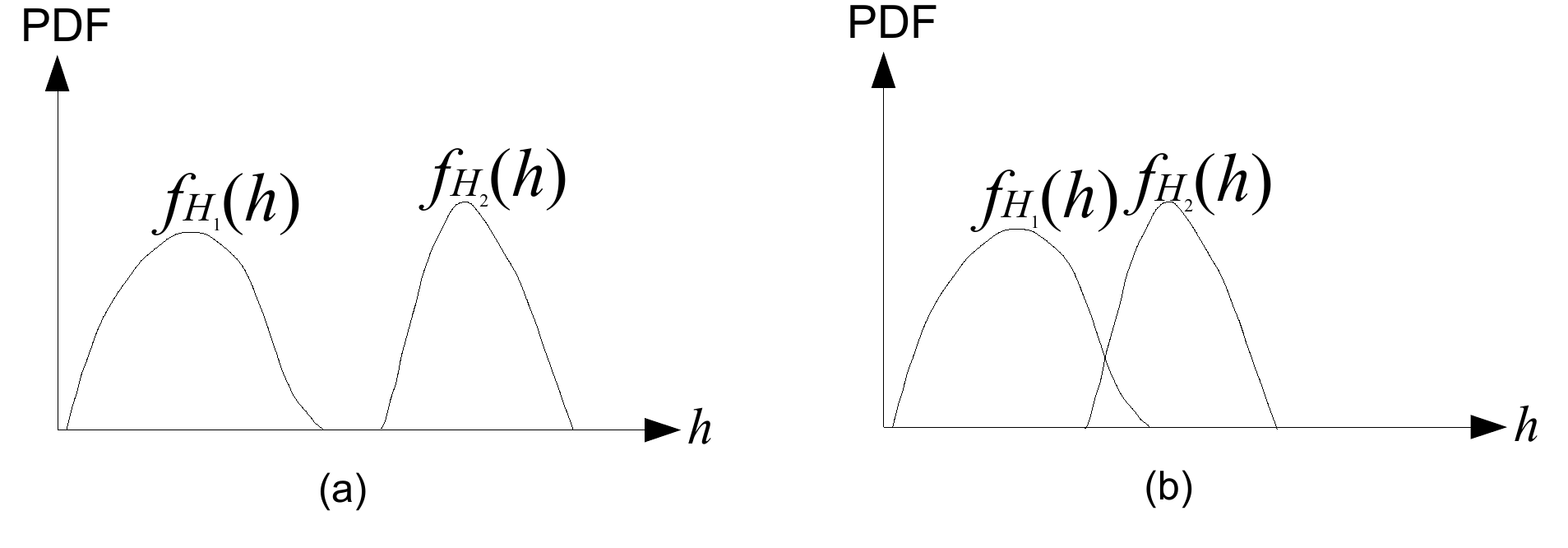, width=0.6\textwidth}
\caption{Two examples of relations between two fading channels.}
\label{Fig_channel_distr}
\end{figure}

In order to find the corresponding capacity region, in this work we resort to \textit{identifying} whether random channel gains in a GMU channel are stochastically \textit{orderable} or not. Stochastically orderable means that there exists an \textit{equivalent} GMU channel\footnote{Here the \textit{equivalent channel} means that it has the same capacity region as the original one.} in which we can reorder channel realizations among different transmitter-receiver pairs such that they satisfy a certain IT order. For example, an \textit{orderable two-user GBC} means that under the same noise distributions at the two receivers, in the equivalent GBC, one channel has realizations of channel gains always larger than the other within a codeword length. We attain this goal mainly by the following elements: \textit{stochastic orders} \cite{shaked_stochastic_order}, \textit{coupling} \cite{Thorisson_coupling}, and the \textit{same marginal property} \cite{Cover_IT_text}. The stochastic orders have been widely used in the last several decades in different areas of probability and statistics
such as reliability theory, queueing theory, and operations
research, etc., see for example \cite{shaked_stochastic_order} and references therein. Different stochastic orders such as the \textit{usual stochastic order}, the \textit{convex order}, and the \textit{increasing convex
order} can help us to identify the location, dispersion,
or both location and dispersion of random variables, respectively \cite{shaked_stochastic_order}.
Properly choosing a stochastic order to compare the
channels may allow us to construct an equivalent channel. In addition to memoryless channels, we also investigate MU channels with memory; in particular the indecomposable finite-state BC (IFSBC) \cite{Dabora_degraded_finite_state_BC}. In the IFSBC model, the channel input-output relation is governed by a state sequence that depends on the channel input, outputs, and previous
states. In addition, the effect of the initial channel state on the state transition probabilities diminishes over time. The concept of finite-state channels has been applied for example to the multiple access channel \cite{Permuter_FSMAC}, degraded BCs \cite{Dabora_degraded_finite_state_BC}, and to the case with feedback \cite{Dabora_indecomposable_finite_state_FBK}.

The main contributions of this paper are summarized as follows:
\begin{enumerate}
\item We construct three schemes for channel comparison under GBC and discuss them:
    \begin{itemize}
    \item We first invoke the concept of maximal coupling to provide an illustrative and easy access to the classification of random channels.
    \item We then exploit coupling by integrating the usual stochastic order and the same marginal property.
    \item In addition to the above schemes which are related to coupling, we explicitly construct a \textit{copula} \cite{Nelson_copulas} and prove the equivalence of coupling and copula in our setting.
\end{itemize}
\item Based on the coupling scheme,
 \begin{itemize}
\item We connect the trichotomy order among channel gains in the constructed equivalent channels to different IT orders, in order to characterize the capacity regions of the GIC and GWTC.
\item We further extend the proposed framework to time-varying channels with memory. In particular, we consider the IFSBC \cite{Dabora_degraded_finite_state_BC} as an example. The Markov fading channel, which is commonly used to model memory effects in wireless channels, is also discussed.
\item Several examples with practical channel distributions are illustrated to demonstrate the applications of the developed framework.
\end{itemize}
\end{enumerate}

Some of the main contributions are summarized in Table \ref{tab_comparison}.\\

\begin {table}
\caption {Comparison of the conditions of different Gaussian MU channels under perfect and statistical CSIT and the related capacity results.} \label{tab_comparison}
\begin{center}
\begin{tabular}{|c|c|c|c|}
  \hline
   &  Conditions under  & Conditions under  &  Capacity results under \\
   & perfect CSIT & statistical CSIT  & statistical CSIT\\\hline
    Degraded GBC&  $h_2\geq h_1$ & $H_2\geq_{st} H_1$  &  \eqref{EQ_GBC_capacity_region} \\\hline
    Strong GIC&    $h_{21}\geq h_{11}$ and $h_{12}\geq h_{22}$  & $H_{21}\geq_{st} H_{11}$ and $H_{12}\geq_{st} H_{22}$  & \eqref{EQ_CMAC} and \eqref{EQ_CMAC_def}\\\hline
    Very strong GIC&  $\frac{h_{21}}{1+P_2h_{22}}\geq h_{11}\mbox{ and }\frac{h_{12}}{1+P_1h_{11}}\geq h_{22}$ & $\frac{H_{21}}{1+P_2H_{22}}\geq_{st} H_{11}\mbox{ and }\frac{H_{12}}{1+P_1H_{11}}\geq_{st} H_{22}$  & \eqref{EQ_CMAC2}  \\\hline
    Degraded GWTC& $h\geq g$  & $H\geq_{st} G$  &  \eqref{EQ_CWTC} \\
  \hline
  Degraded IFSBC (Markov fading)& \eqref{EQ_stochastic_degraded}  & \eqref{EQ_Markov_condition1}, \eqref{EQ_Markov_condition2}, and \eqref{EQ_Markov_condition3}   &  \cite{Dabora_degraded_finite_state_BC} or \eqref{EQ_IFSBC_capacity} \\
  \hline
\end{tabular}
\end{center}
\end {table}
The remainder of the paper is organized as follows. In Section \ref{Sec_multiuser_channels}, we formulate an abstract problem and develop our framework based on maximal coupling, coupling, and copulas for fading GBC with statistical CSIT. We then apply this framework to fading GIC and GWTC. In Section \ref{Sec_memory}, we discuss the IFSBC as an application to channels with memory. Finally, Section \ref{Sec_conclusion} concludes the paper.

\emph{Notation}: Upper case normal/bold letters denote random
variables/random vectors (or matrices), which will be defined
when they are first used; lower case normal/bold letters denote
deterministic variables/vectors. Vector $\bm a\triangleq [a_i,\,a_{i+1},\,\cdots,\,a_{j-1},\,a_{j}]$ is interchangeably denoted by $a_i^j$ while $a_1^j$ is simplified as $a^j$. A diagonal matrix formed by a vector $\bm a$ is denoted by diag$\{\bm a\}$. Uppercase calligraphic letters denote sets. The expectation is denoted by $\mathds{E}[\cdot]$. We denote the probability mass function (PMF) and probability density function (PDF) of a random variable $X$ by $p_X$ and $f_X$, respectively. The probability of event $A$ is denoted by Pr$(A)$. The cumulative distribution function
(CDF) is denoted by $F_X(x)$, where $\bar{F}_X(x)=1-F_X(x)$ is the
complementary CDF (CCDF) of $X$. $\,X\sim \,F$ means that the random variable $X$ follows the
distribution with CDF $F$. The mutual information between two random variables $X$ and $Y$ is denoted by $I(X;Y)$ while the conditional mutual information given $Z$ is denoted by $I(X;Y|Z)$. The differential and conditional differential entropies are denoted by $h(\cdot)$ and $h(\cdot|\cdot)$, respectively. A Markov chain relation between $X$, $Y$, and $Z$ is denoted by $X - Y - Z$. $\mathrm{Unif}(a,b)$ denotes the uniform distribution between $a\in\mathds{R}$ and $b\in\mathds{R}$ and $\mathds{N}^0=\{0,\,\mathds{N}\}$ is the set of non-negative integers. The Bernoulli distribution with probability $p$ is denoted by Bern$(p)$. The support of a random variable $X$ is interchangeably denoted by supp$(X)$ or supp$(f_X)$. The logarithms used in the paper are all with respect to base 2. We define $C(x)\triangleq\frac{1}{2}\log(1+x)$. We denote equality in distribution by $=_d$. The convolution of functions $f$ and $g$ is denoted by $f\ast g$. Circularly symmetric complex AWGN with zero mean and variance $P$ is denoted by $\mathcal{CN}(0,P)$. The convex hull of a set $\mathcal{A}$ is denoted by $co(\mathcal{A})$.

\section{Problem Formulation and Main Results}\label{Sec_multiuser_channels}
In this section, we first introduce the problem formulation and then develop a framework to classify fading GBCs such that we are able to obtain the corresponding ergodic capacity results under statistical CSIT. After that we apply the coupling scheme to GIC and GWTC. In brevity, the underlying spirit of all schemes is: while keeping the distributions of marginal channels fixed, we change the joint distribution of the GMU in such a way that it has a certain special structure
which allows us to obtain the capacity results. In this paper we assume that each node in the considered GMU channels is equipped with a single antenna.

\subsection{Problem Formulation and Preliminaries}

As motivated above, we formulate the problem statement as follows. We have to find two sets $\mathcal{A}$ and $\mathcal{B}$: set $\mathcal{A}$ is a subset of all fading channel gains of a particular GMU, e.g., $\mathcal{A}\subset \{(H_1,\,H_2)\}$ for a two-user GBC where $H_1$ and $H_2$ are the fading channel gains to the first and second users, respectively; set $\mathcal{B}$ is composed of channel gains of an equivalent GMU. Intuitively, the sets $\mathcal{A}$ and $\mathcal{B}$ shall possess the following properties:
\begin{itemize}
 \item[$\bm {P1.}$] Channel gains in set $\mathcal{B}$ lead to (existing) capacity results, which should be the same as those of the original channels in set $\mathcal{A}$, and may follow certain IT orders.
 \item[$\bm {P2.}$] There exists a constructive way to find a mapping $f:\,\mathcal{A}\mapsto \mathcal{B}$;
\end{itemize}

The considered problem in this work is formulated in an abstract representation as follows, also illustrated in Fig. \ref{Fig_framework}.\\

\textbf{Problem 1:} Under statistical CSIT, find a set of tuples
\begin{align}
\mathcal{S}=\{&(\mathcal{A},\,f):\,\,\,\,f:\,\mathcal{A}\mapsto \mathcal{B},\label{Def_design_goal1}\\
&a\in\mathcal{A} \mbox{ and }b\in\mathcal{B}\mbox{ have the same marginal distributions, such that the capacities of the underlying}\notag\\
&\mbox{channels in $\mathcal{A}$ and $\mathcal{B}$ are the same.}\}\label{Def_design_goal2}
\end{align}

The same marginal property provides us the degree of freedom to construct equivalent channels in which the realizations of all random channel tuples are possible to be aligned in a desired IT order, while this alignment can be achieved by several schemes discussed in the following. Note that the choices of the sets $\mathcal{A}$ and $\mathcal{B}$ depend on the topologies of the MU channels, which will also be explained in the following case by case.
\begin{figure}[h]
\centering \epsfig{file=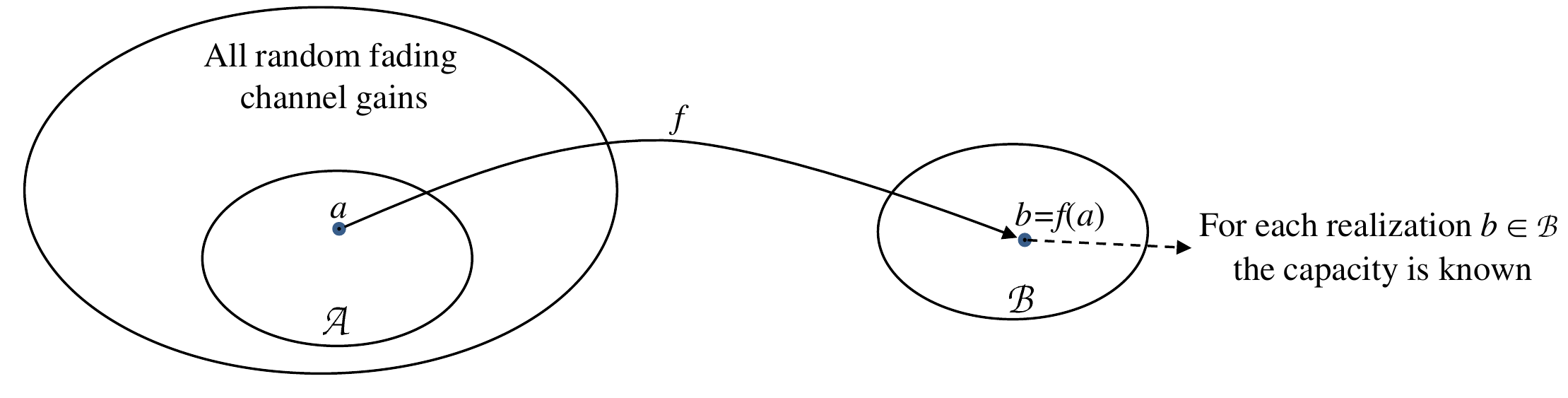, width=0.8\textwidth}
\caption{The proposed scheme identifies the ergodic capacity regions under statistical CSIT.}
\label{Fig_framework}
\end{figure}

\begin{rmk}\label{RMK:same_marginal_wiretap}\normalfont
The optimal classification identifies the three elements of a tuple $(\mathcal{A},f,\mathcal{B})$, simultaneously, instead of fixing $\mathcal{B}$ and then finding $(\mathcal{A},f)$. However, this way may result in configurations for which new inner and outer bounds have to be established, which is out of the scope of this work. Therefore, we restrict to those $\mathcal{B}$ for which capacities are known.
\end{rmk}

Some important definitions related to our solutions to \textbf{Problem 1} are shown in the following.

\begin{definition}\cite[(1.A.3)]{shaked_stochastic_order}\label{shaked_stochastic_order}
For random variables $X$ and $Y$, $X$ is smaller than $Y$ in the \textit{usual stochastic order}, denoted as, $X\leq_{st} Y$, if and only if
$\bar{F}_{X}(x)\leq \bar{F}_{Y}(x)$ for all $x\in(-\infty,\infty)$.
\end{definition}

Note that Definition \ref{shaked_stochastic_order} is applicable to both discrete or continuous random variables.

\begin{definition}\cite[Definition 2.1]{Ross_2nd_course}\label{Def_coupling}
The pair $(\tilde{X},\,\tilde{Y})$ is a coupling of the random variables $(X,Y)$ if $\tilde{X}=_{d}X$ and $\tilde{Y}=_{d}Y$.
\end{definition}

\begin{definition}\label{Def_{H_2}D_copula}
A two-dimensional copula \cite[(2.2.2a), (2.2.2b), (2.2.3)]{Nelson_copulas} is a function $C_0:\,[0,1]^2\mapsto [0,1]$ with the following properties:
\begin{enumerate}
\item For every $u,\,v\in[0,1]$,
\begin{align}
C_0(u,0)=0=C_0(0,v),\label{EQ_2.2.a}\\
C_0(u,1)=u \mbox{ and }C_0(1,v)=v;\label{EQ_2.2.b}
\end{align}
\item For every $u_1,\,u_2,\,v_1,\,v_2\in[0,1]$ such that $u_1\leq u_2$ and $v_1\leq v_2$,
\begin{align}\label{EQ_2.2.3}
C_0(u_2,v_2)-C_0(u_2,v_1)-C_0(u_1,v_2)+C_0(u_1,v_1)\geq 0.
\end{align}
\end{enumerate}
\end{definition}

\begin{definition}\cite[Section 2.2]{Ross_2nd_course}
For the random variables $(X,\,Y)$, the coupling $(\tilde{X},\,\tilde{Y})$ is called a maximal coupling if $\textnormal{Pr}(\tilde{X}=\tilde{Y})$ gets its maximal value among all the couplings of $(X,\,Y)$.
\end{definition}

\subsection{Fading Gaussian Broadcast Channels with Statistical CSIT}
The capacity of the degraded BC is known \cite{Bergmans_degraded_BC}. For non-degraded BCs, only the inner and outer bounds are known, e.g., Marton's inner bound \cite{Marton_inner_bound} and
Nair-El Gamal's outer bound \cite{Nair_outer_bound}. Therefore, it shall be
easier to characterize the capacity region of a GBC under statistical CSIT if we can identify that it is degraded\footnote{In the following, we call a stochastically degraded channel simply a degraded channel due to the same marginal property.}.
Denote the CCDFs of the random gains $H_1$ and $H_2$ in a two-user GBC by $\bar{F}_{H_1}$ and $\bar{F}_{H_2}$, respectively. Receiver $k$'s signal can be stated as
\begin{align}
Y_k=\sqrt{H_k} X+N_k,\,k=1,\,2,\label{EQ_channel_model1}
\end{align}
$X$ is the channel input with an input power constraint $\mathds{E}[X^2]\leq P_T$. The noises
$N_1$ and $N_2$ at the corresponding receivers are independent AWGN with zero mean and unit variance. We assume that there is perfect CSIR such that the receivers can compensate the phase rotation of their own channels, respectively, without changing the capacity. Then we can focus on a real GBC instead of a complex one. We also assume that the transmitter only knows the statistics but not the instantaneous realizations of $H_1$ and $H_2$.

For random variables whose PDFs intersect as in Fig. \ref{Fig_channel_distr} (b), a straightforward idea to reorder $H_1$ and $H_2$ is: we try to form a coupling $(\tilde{H}_1,\,\tilde{H}_2)$ of $(H_1,\,H_2)$ with marginal PDFs $\tilde{f}_{H_1}=f_{H_1}$ and $\tilde{f}_{H_2}=f_{H_2}$, respectively, where we have the reordered realizations of $\tilde{H}_1$ and $\tilde{H}_2$ such that Pr($\tilde{h}_1=\tilde{h}_2)=\int\min\left\{\tilde{f}_{H_1}(h),\,\tilde{f}_{H_2}(h)\right\}dh$ if $h\in\,\mbox{supp}(f_{H_1})\cap\mbox{supp}(f_{H_2})$. Otherwise, $\tilde{H}_1$ and $\tilde{H}_2$ follow the original PDFs $\tilde{f}_{H_1}$ and $\tilde{f}_{H_2}$, respectively. If this reordering is feasible, we know that one of the fading channel is always equivalently not weaker than the other one. However, does such a reordering exist? The answer is yes, and it relies on the concept of \textit{maximal coupling} \cite{Ross_2nd_course}. In addition to the maximal coupling scheme, we also propose other solutions for \textbf{Problem 1}, which are related to coupling and copulas \cite{Nelson_copulas} and all of them are summarized as follows:
\begin{enumerate}
\item The first scheme is $(\mathcal{A},f)$, achieved by maximizing $\mbox{Pr(}\tilde{H}_1=\tilde{H}_2)$, where realizations of $\tilde{H}_1$ and $\tilde{H}_2$ belonging to the intersection of the domains of $\tilde{f}_{H_1}$ and $\tilde{f}_{H_2}$ are exactly aligned, i.e., $\tilde{h}_1=\tilde{h}_2$.
\item The second scheme is $(\mathcal{A},f')$, achieved by constructing an explicit coupling that not necessarily maximizes $\mbox{Pr(}\tilde{H}_1=\tilde{H}_2)$ but each realization of channel gain pairs preserves the same trichotomy order.
\item The third scheme is $(\mathcal{A},f'')$, achieved by explicitly constructing a copula \cite{Nelson_copulas} (a new joint distribution between the fading channels), such that each realization of channel pairs has also the same trichotomy order. Note that the existence of the copula is proved by Sklar's theorem \cite{Nelson_copulas}.
\end{enumerate}

At the end, we will prove that the schemes by coupling and copulas are equivalent. We derive the two-receiver case in detail. Without loss of generality, we assume that channel 1 is degraded with respect to channel 2.
\newpage
\begin{Theorem}\label{Th_all}
\normalfont
For a two-user fading GBC with statistical CSIT, assume that the PDFs of $H_1$ and $H_2$, namely, $f_{H_1}$ and $f_{H_2}$, are continuous. Then the following selections of $\mathcal{S}$ lead to degraded GBCs in $\mathcal{B}$:
\begin{align}
\hspace{-1.4cm}\mathcal{A}=&{\normalfont \big\{ ( H_1,\, H_2 ) : H_{2}\geq_{st}H_{1}\big\}},\label{EQ_coupling_construction}\\
\hspace{-1.4cm}f(\mathcal{A})=&\{(\tilde{H}_1,\,\tilde{H}_2 ): \,\,k=1,2,\notag
\end{align}
\begin{numcases}{\hspace{3cm}\tilde{f}_{H_k}(h)=}
                     \frac{f_{H_k}(h)-f_{min}(h)}{1-p},\!\mbox{ if }V=0,\label{EQ_max_coupling1} \\
                     \frac{f_{min}(h)}{p}\,\,\,\,\,\,\,\,\,\,\,\,\,\,\,\,\,\,\,\,\,\,\,\,\,\,\,\,,\mbox{ otherwise},\label{EQ_max_coupling2}
\end{numcases}
\begin{align}
&\hspace{0.5cm}\},\notag\\
\hspace{5cm}f'(\mathcal{A})&=\{(\tilde{H}_1,\,\tilde{H}_2):\,\tilde{H}_1= F_{H_1}^{-1}(U),\,\tilde{H}_2= F_{H_2}^{-1}(U),\,U\sim\mbox{Unif}(0,1)\},\label{EQ_construct_f}\\
\hspace{5cm}f''(\mathcal{A})&=\{(\tilde{H}_1,\,\tilde{H}_2)\!:\, \bar{F}_{\tilde{H}_1,\tilde{H}_2}(\tilde{h}_1,\tilde{h}_2)=\min\{\bar{F}_{H_1}(\tilde{h}_1),\bar{F}_{H_2}(\tilde{h}_2)\}\notag\\
&\hspace{2.3cm}\mbox{    is the Fr\'{e}chet-Hoeffding upper bound of copulas}
\!\},\label{EQ_new_joint_CCDF}\\
&\hspace{-0.3cm}\mathcal{B}=\{(\tilde{H}_1,\,\tilde{H}_2)\in \{f(\mathcal{A}),\, f'(\mathcal{A}),\, f''(\mathcal{A})\}\},\notag
\end{align}
where $f_{min}(h)\triangleq \min(f_{H_1}(h),\, f_{H_2}(h))$, $V\sim \mbox{Bern}(p)$ is independent of all other random variables,  $p\triangleq\int_{-\infty}^{\infty}f_{min}(h)dh$, $h\in{\normalfont \mbox{ supp}}(H_1)\cup{\normalfont\mbox{ supp}}(H_2)$. The mapping $f'$ is equivalent to $f''$.
\end{Theorem}
\begin{proof}
The proof is relegated to \ref{APP_all}.
\end{proof}

Based on \cite{Bergmans_degraded_BC}, the capacity region $\mathcal{R}$ of a degraded fading GBC under statistical CSIT satisfying $V-X-Y_2-Y_1$ can be expressed as:
\begin{align}\label{EQ_GBC_capacity_region}
\mathcal{R}\triangleq \bigcup_{f_{VX}, \mathds{E}[X^2]\leq P_T} \left\{\begin{array}{l}
                                         (R_1,\,R_2):\,\, R_1\leq I(V;Y_1|H_1),\,\,                                        R_2\leq I(X;Y_2|V,H_2)
                                       \end{array}\right\}.
\end{align}

\begin{rmk}\normalfont
Solving the optimal $f_{VX}$ in \eqref{EQ_GBC_capacity_region} is still an open problem. In \cite{Abbe_non_gaussian}, it is shown that with statistical CSIT, a locally perturbed Gaussian input can be better than the Gaussian one.
\end{rmk}

A generalization of the set $\mathcal{A}$ in \eqref{EQ_coupling_construction} to more users can be easily derived as follows.

\begin{Corollary}\normalfont\label{coro_N_user}
For a $K$-user fading GBC with statistical CSIT, $K\geq 2$, if
$H_K\geq_{st}H_{K-1}\geq_{st}\cdots\geq_{st}H_1$, then it is degraded.
\end{Corollary}

\begin{rmk}
Since the usual stochastic order is closed under convolution \cite[Theorem
1.A.3]{shaked_stochastic_order}, we can easily extend the discussions above to \textit{clusters of scatterers} \cite{Tse_scattering}.
Furthermore, the numbers of clusters of scatterers can even be random variables as well and the discussions above hold as long as those numbers follow the usual stochastic order \cite[Theorem
1.A.4]{shaked_stochastic_order}, again.
\end{rmk}

\begin{rmk}\normalfont
The three schemes in Theorem \ref{Th_all} are similar in the sense that, given the marginal distributions, we aim to form joint distributions to achieve a degraded GBC. Contrary to copulas, schemes by coupling and maximal coupling do not directly construct an explicit joint distribution. Instead, equivalent marginal distributions are explicitly constructed, i.e., \eqref{EQ_max_coupling1}, \eqref{EQ_max_coupling2}, and \eqref{EQ_construct_f}, which implicitly determines the joint distributions. The scheme from copulas, in contrast, explicitly constructs a joint distribution from the marginal distributions. Note that maximal coupling and coupling reorder the fading channel realizations differently. More specifically, maximal coupling reorders realizations belonging to the intersection of PDFs in a strict sense, i.e., $\tilde{h}_1=\tilde{h}_2$. In contrast, coupling by \eqref{EQ_construct_f}, in general, does not reorder part of the realizations as $\tilde{h}_1=\tilde{h}_2$, but only results in $\tilde{h}_2>\tilde{h}_1$ when $F_{H_k}(\tilde{h}_k)\in(0,1),\,k=1,\,2$, which can be easily observed from \eqref{EQ_construct_f}.
\end{rmk}

\begin{Ex}
Consider a three-user fading GBC whose channel magnitudes are independent Nakagami-$m$ distributed with \textit{shape parameters} $m_1$, $m_2$, and $m_3$, and \textit{spread parameters} $w_1$, $w_2$, and $w_3$, respectively. From Corollary \ref{coro_N_user}, we know that this channel is degraded if
\begin{align}
\frac{\gamma\left(m_3,\frac{m_3 x}{w_3}\right)}{\Gamma(m_3)}\geq \frac{\gamma\left(m_2,\frac{m_2 x}{w_2}\right)}{\Gamma(m_2)}\geq \frac{\gamma\left(m_1,\frac{m_1 x}{w_1}\right)}{\Gamma(m_1)},\,\,\forall x,\notag
\end{align}
where $\gamma(s,x)=\int_0^xt^{s-1}e^{-t}dt$ is the incomplete gamma function and $\Gamma(s)=\int_0^{\infty}t^{s-1}e^{-t}dt$ is the ordinary gamma function \cite[p. 255]{Stegun_handbook}. An example satisfying the above inequality is $(m_1,w_1)=(0.5,1)$, $(m_2, w_2)=(1,2)$, and $(m_3, w_3)=(1,3).$
\end{Ex}

{\begin{rmk}\normalfont
Note that the probability of $\tilde{H}_1\neq\tilde{H}_2$ from the maximal coupling $(\tilde{H}_1,\,\tilde{H}_2)$ of $(H_1,\,H_2)$ can be described by the total variation distance between $H_1$ and $H_2$ \cite[Proposition 2.7]{Ross_2nd_course} as
\begin{align}
\mbox{Pr(}\tilde{H}_1\neq\tilde{H}_2)=1-\int_{-\infty}^{\infty}\min(f_{H_1}(h),f_{H_2}(h))dh= d_{TV}(H_1,H_2),
\end{align}
where $d_{TV}(X,Y)= \underset{\mathcal{A}}{\sup}|\mbox{Pr(}X\in \mathcal{A})-\mbox{Pr(}Y\in \mathcal{A})|$ is the total variation distance. It is consistent to the intuition from \eqref{EQ_max_coupling1} and \eqref{EQ_max_coupling2} that, if $d_{TV}(H_1,H_2)$ is larger (or equivalently, the overlapping area of $\tilde{f}_{H_1}$ and $\tilde{f}_{H_2}$ is smaller), then $\mbox{Pr(}\tilde{H}_1\neq\tilde{H}_2)$ is smaller.
\end{rmk}}

\begin{rmk}\normalfont
Note that when the statistical CSIT is considered instead of the perfect CSIT, the expression of usual stochastic order makes it a good generalization of the trichotomy order commonly used in the expressions of IT orders, i.e., we just replace the operator $\geq$ by $\geq_{st}$. However, the reordering operation done in the equivalent channel may be too strict. This is because the considered IT channel orders are in fact described solely by (conditional) distributions but not by the realizations of individual channels.
\end{rmk}



{
\begin{rmk} \label{RMK_copula}\normalfont
We can also directly identify that \eqref{EQ_new_joint_CCDF} is a \textit{two-dimensional copula} \cite[Definition 2.2.2]{Nelson_copulas} by Definition \ref{Def_{H_2}D_copula} without the aid of Fr\'{e}chet-Hoeffding-like upper bound. This identification is relegated to \ref{APP_copula}.
\end{rmk}}

\subsection{Fading Gaussian Interference Channels with Statistical CSIT}\label{Sec_SISO}
In this section, we identify sufficient conditions to obtain the capacity regions of equivalent two-user Gaussian interference channels with strong and very strong interferences by the proposed scheme. Examples illustrate the results.
We assume that each receiver perfectly knows the two channels and the received signals can be stated as
\begin{align}
Y_1&=\sqrt{H_{11}}e^{j\Phi_{11}} X_1+\sqrt{H_{12}}e^{j\Phi_{12}} X_2+N_1\triangleq \tilde{H}_{11}X_1+\tilde{H}_{12}X_2+N_1,\label{EQ_channel_model1_1}\\
Y_2&=\sqrt{H_{21}}e^{j\Phi_{21}} X_1+\sqrt{H_{22}}e^{j\Phi_{22}} X_2+N_2\triangleq \tilde{H}_{21}X_1+\tilde{H}_{22}X_2+N_2,\label{EQ_channel_model1_2}
\end{align}
where $H_{kj}$ and $\Phi_{kj}\in [0,\,2\pi]$ are real-valued non-negative mutually independent random
variables denoting the channel gain and the phase of the fading channel between the $j$-th transmitter to the $k$-th receiver, respectively, where $k,\,j\in\{1,\,2\}$. The CCDF of $H_{kj}$ is denoted by
$\bar{F}_{H_{kj}}$. The channel inputs at the transmitters 1 and 2 are denoted by $X_1$ and $X_2$, respectively. We consider the channel input power constraints
$\mathds{E}[|X_1|^2]\leq P_1$ and $\mathds{E}[|X_2|^2]\leq P_2$, respectively. Noises
$N_1\sim\mathcal{CN}(0,1)$ and $N_2\sim\mathcal{CN}(0,1)$ at the corresponding receivers are independent. We
assume that the transmitters only know the statistics but not the
instantaneous realizations of $\{H_{kj}\}$ and $\{\Phi_{kj}\}$. Hence, the channel input signals $\{X_j\}$ are not functions of the channel realizations and therefore are independent to $\{H_{kj}\}$ and $\{\Phi_{kj}\}$. In addition, without loss of generality, we assume that channel gains, channel phases, and noises are mutually independent. We assume perfect CSIR. However, we are not able to simplify the channel into a real one as in the GBC case even with perfect CSIR due to signals from different transmitters encountering different fadings.
We derive two results for two-user GICs with only statistical CSIT as follows, which correspond to the stochastic version of strong and very strong interferences, respectively.
\begin{Theorem}\normalfont\label{Th_cond_strong_IC}
The selection of the following set $\mathcal{S}$ leads to GICs with strong interference in $\mathcal{B}$:
\begin{align}
\mathcal{A}&=\{(H_{11},\,H_{12},\,H_{21},\,H_{22}):\,H_{21}\geq_{st}H_{11}\mbox{ and }H_{12}\geq_{st}H_{22};\,H_{jk},\,j,\,k=1,\,2 \mbox{ are mutually independent}\},\label{EQ_strong_IC_constraint}\\
f'(\mathcal{A})&=\left\{(H_{11}',\,H_{12}',\,H_{21}',\,H_{22}'):\,H_{21}' = F_{H_{21}}^{-1}(U_1),\,H_{11}' = F_{H_{11}}^{-1}(U_1),\,U_1\sim\mbox{Unif}(0,1),\right.\,\,\notag\\
&\hspace{4.6cm}H_{12}' = F_{H_{12}}^{-1}(U_2),\,H_{22}' = F_{H_{22}}^{-1}(U_2),\,U_2\sim\mbox{Unif}(0,1),\notag\\
&\left.\hspace{4.5cm}\,U_1\mbox{ is independent of }U_2\right\},\label{EQ_strong_IC_mapping}\\
\mathcal{B}&=\{(H_{11}',\,H_{12}',\,H_{21}',\,H_{22}')\in f'(\mathcal{A})\},
\end{align}
and the corresponding ergodic capacity region is as:
\begin{align}\label{EQ_CMAC}
\mathcal{C}(P_1,P_2)=\mathcal{C}_1(P_1)\cap\mathcal{C}_2(P_2),
\end{align}
where
\begin{align}\label{EQ_CMAC_def}
\mathcal{C}_j(P_j)\triangleq\big\{(R_1,R_2):\,R_1&\leq \mathds{E}[C(H_{j1}P_1)],\notag\\
R_2&\leq \mathds{E}[C(H_{j2}P_2)],\notag\\
R_1+R_2&\leq \mathds{E}[C(H_{j1}P_1+H_{j2}P_2)]\big\},\,j=1,\,2.
\end{align}
\end{Theorem}
\begin{proof}
The proof is relegated to \ref{APP_Th_cond_strong_IC}.
\end{proof}


%
%

For the GIC with very strong interference, we introduce two feasible sets $\mathcal{S}_1\triangleq(\mathcal{A}_1,\,f_1'(\mathcal{A}_1))$ and $\mathcal{S}_2\triangleq(\mathcal{A}_2,\,f_2'(\mathcal{A}_2))$, to \textbf{Problem 1}. In $\mathcal{A}_1$ we assume that the channel gains are mutually independent, while in $\mathcal{A}_2$ the channel gain are allowed to have a specific correlation structure.
\newpage
\begin{Theorem}\normalfont\label{Th_cond_very_strong_IC}
Define $Z_1\triangleq\frac{H_{21}}{1+P_2H_{22}}$ and $Z_2\triangleq\frac{H_{12}}{1+P_1H_{11}}$. The selections of the following sets $\mathcal{S}_1$ and $\mathcal{S}_2$ lead to GICs with very strong interference in $\mathcal{B}$:
\begin{align}
\mathcal{S}_1=\Bigg\{&(\mathcal{A}_1,\,f_1'):\notag\\
&\mathcal{A}_1=\left\{(H_{11},\,H_{12},\,H_{21},\,H_{22}):\,Z_1\geq_{st} H_{11}\mbox{ and }Z_2\geq_{st} H_{22},\,H_{jk},\,{j,\,k=1,\,2}\mbox{ are mutually independent}\right\},\label{Constraint_SMP_very_strong2}\\
&f_1'(\mathcal{A}_1)=\left\{(H_{11}',\,H_{12}',\,H_{21}',\,H_{22}'):\,H_{12}' =
(1+P_2H_{22}')F_{Z_1}^{-1}(U_1),\,H_{11}'=F_{H_{11}}^{-1}(U_1),\,U_1\sim\mbox{Unif}(0,1),
\right.\notag\\
&\left.\hspace{4.6cm}  H_{21}' = (1+P_1H_{11}')F_{Z_2}^{-1}(U_2),\,H_{22}'=F_{H_{22}}^{-1}(U_2),\,U_2\sim\mbox{Unif}(0,1) \right\}\Bigg\},\label{EQ_VS_f1_prime_2nd}\\
\mathcal{S}_2=\Bigg\{&(\mathcal{A}_2,\,f_2'):\notag\\
&\mathcal{A}_2=\Big\{(H_{11},\,H_{12},\,H_{21},\,H_{22}):\,Z_1\geq_{st} H_{11}\mbox{, }Z_2\geq_{st} H_{22},\mbox{ and for all }a,\,b\geq 0,\label{Constraint_very_strong}\\
&\hspace{3cm}F_{Z_1,\,H_{22}}(a,b)=\min\left\{F_{Z_1}(a),\,F_{H_{22}}(b)\right\},\,F_{Z_2,\,H_{11}}(a,b)=\min\left\{F_{Z_2}(a),\,F_{H_{11}}(b)\right\}\Big\},\label{Constraint_SMP_very_strong}\\
&f_2'(\mathcal{A}_2)=\left\{(H_{11}',\,H_{12}',\,H_{21}',\,H_{22}'):\,H_{12}' =
(1+P_2H_{22}')F_{Z_1}^{-1}(U),\,H_{11}'=F_{H_{11}}^{-1}(U),
\right.\notag\\
&\left.\hspace{4.6cm}  H_{21}' = (1+P_1H_{11}')F_{Z_2}^{-1}(U),\,H_{22}'=F_{H_{22}}^{-1}(U),\,U\sim\mbox{Unif}(0,1) \right\}\Bigg\},\\
&\hspace{-1.2cm}\mathcal{B}=\{(H_{11}',\,H_{12}',\,H_{21}',\,H_{22}')\in f_1'(\mathcal{A}_1)\mbox{ or }f_2'(\mathcal{A}_2)\},
\end{align}
and the corresponding ergodic capacity region is as:
\begin{align}\label{EQ_CMAC2}
\mathcal{C}(P_1,P_2)=\big\{(R_1,R_2):\,R_1&\leq \mathds{E}[C(H_{11}P_1)],\notag\\
R_2&\leq \mathds{E}[C(H_{22}P_2)]\big\}.
\end{align}
\end{Theorem}
\begin{proof}
The proof is relegated to \ref{APP_Th_cond_very_strong_IC}.
\end{proof}

\begin{rmk}\normalfont
It can be observed that the conditions in \eqref{EQ_strong_IC_constraint} and \eqref{Constraint_very_strong}, \eqref{Constraint_SMP_very_strong} are generalized from the trichotomy orders in the GIC under strong and very strong interference constraints with perfect CSIT, respectively, to the stochastic orders. This means that, perfect CSIT is equivalent to having a Dirac delta function and then knowing the statistics is identical to knowing the perfect CSIT. Therefore, \eqref{EQ_strong_IC_constraint} and \eqref{Constraint_very_strong}, \eqref{Constraint_SMP_very_strong} cover the strong and very strong interference conditions of perfect CSIT cases, respectively.
\end{rmk}

\begin{rmk}\normalfont
A recent work \cite{Vahid_binary_fading_IC} proves the capacity region of a two-user binary fading interference channel with statistical CSIT under both weak and strong interferences, where the received signals are given by
\begin{align}
Y_k=H_{kk}X_k\oplus H_{k\bar{k}}X_{\bar{k}},\,k=1,\,2, \label{EQ_Vahid_IC}
\end{align}
where $\bar{k}=3-k$, $H_{kk},\,H_{k\bar{k}}\in\{0,\,1\}$ and all algebraic operations are in the binary field. Our result in Theorem \ref{Th_cond_strong_IC} and also the discussion in Section \ref{Sec_multiuser_channels} are valid for the binary fading IC \eqref{EQ_Vahid_IC}. More specifically, let $H_{11}\sim \mbox{Bern}(p_d)$, $H_{22}\sim \mbox{Bern}(p_d)$, $H_{12}\sim \mbox{Bern}(p_c)$, and $H_{21}\sim \mbox{Bern}(p_c)$ and the strong interference channel in \cite{Vahid_binary_fading_IC} is defined as $p_d\leq p_c\leq 1$. From the CCDFs it is easy to see that $H_{21}\geq_{st} H_{11}$. Similarly, we can find $H_{12}\geq_{st} H_{22}$. This fact manifests that \eqref{EQ_strong_IC_constraint} is a more general expression of the strong interference condition under statistical CSIT, i.e., it can be used for either discrete or real valued random channels. {Cases in \cite{Vahid_binary_fading_IC} achieving capacity regions but not being covered in our current paper can be considered as future works.}
\end{rmk}

In the following we provide examples to show scenarios in which the sufficient conditions in Theorem \ref{Th_cond_strong_IC} and Theorem \ref{Th_cond_very_strong_IC} are feasible.

\begin{Ex}: Assume that the squares of the four channel gains follow exponential distributions, i.e., $H_{jk}\sim \mbox{Exp}(1/\sigma_{jk}^2)$, $j,\, k=1,2$. From its CCDF, it is easy to see that if $\sigma_{21}^2\geq\sigma_{11}^2$ and $\sigma_{12}^2\geq\sigma_{22}^2$, then the two constraints in \eqref{EQ_strong_IC_constraint} are fulfilled.
\end{Ex}

\begin{Ex}:
Here we show an example for $(\mathcal{A}_1,\,f_1')$ in Theorem \ref{Th_cond_very_strong_IC}. We first find the distributions of the two ratios of random variables in \eqref{Constraint_very_strong}. From \cite[(16)]{Hassibi_ratio_quadratic_form} we can derive
\begin{align}
F_{Z_1}(h)&\overset{(a)}=\textsf{u}(h)-\Sigma_{i=1}^2\frac{\lambda_i^2}{\Pi_{l\neq i}(\lambda_i-\lambda_l)}\frac{1}{|\lambda_i|}e^{\frac{-h}{\lambda_i}}\textsf{u}\left(\frac{h}{\lambda_i}\right)\notag\\
&\overset{(b)}=\! \textsf{u}(h)-\frac{1}{\sigma_{21}^2+hP_2\sigma_{22}^2}e^{\frac{-h}{\sigma_{21}^2}}\textsf{u}\left(\!\frac{h}{\sigma_{21}^2}\!\right),\label{EQ_F_Z}
\end{align}
where in (a), $\textsf{u}(h)$ is the unit step function, $\{\lambda_i\}$ are the eigenvalues of ${\normalfont \mbox{diag}}\{[\sigma_{21}^2,\,\,0]\}-h\cdot P_2\cdot {\normalfont \mbox{diag}}\{[0,\,\,\sigma_{22}^2]\}$, i.e., $\{\lambda_i\}=\{\sigma_{21}^2,\,-h\cdot P_2\cdot\sigma_{22}^2\}$; in (b) we substitute the two eigenvalues into the right hand side (RHS) of (a).
Then we evaluate the first constraint in \eqref{Constraint_very_strong} by checking the difference of CCDFs of $Z_1$ and $H_{11}$, i.e., $\bar{F}_{Z_1}(h)-\bar{F}_{H_{11}}(h)=1-F_Z(h)-e^{\frac{-h}{\sigma_{11}^2}}$, numerically. In the first comparison, we fix the variances of the cross channels as $c\triangleq\sigma_{12}^2=\sigma_{21}^2=1$ and the transmit powers $P\triangleq P_1=P_2=1$, and set the variances of the dedicated channels as $a\triangleq\sigma_{11}^2=\sigma_{22}^2=0.1,\,0.3,\,0.5,\mbox{ and }\,0.7$. Since the conditions in \eqref{Constraint_very_strong} are symmetric and the considered settings are symmetric, once the first condition in \eqref{Constraint_very_strong} is valid, the second one will be automatically valid. The results are shown in Fig. \ref{Fig_very_strong_diff_a} where the vertical axis is the difference between the CCDFs of $Z_1$ and $H_{11}$. From Fig. \ref{Fig_very_strong_diff_a} we can observe that only the values $a=0.1,\,0.3,$ and $0.5$ result in $(H_{11},\,H_{12},\,H_{21},\,H_{22})$ satisfying \eqref{Constraint_very_strong}, while the difference of the CCDFs is negative when $H$ approaches zero under $P=100$. We also investigate the effect of the transmit power constraints to the validity of the sufficient condition in \eqref{Constraint_very_strong}. We consider the case with $a=0.1,\,c=1$ with $P=1,\,10,\,50,$ and $100$ (in linear scale). From Fig. \ref{Fig_very_strong_diff_P} we can observe that only $P=1,\,10,$ and $50$ satisfy \eqref{Constraint_very_strong}.
\end{Ex}

\begin{figure}[h]
\centering \epsfig{file=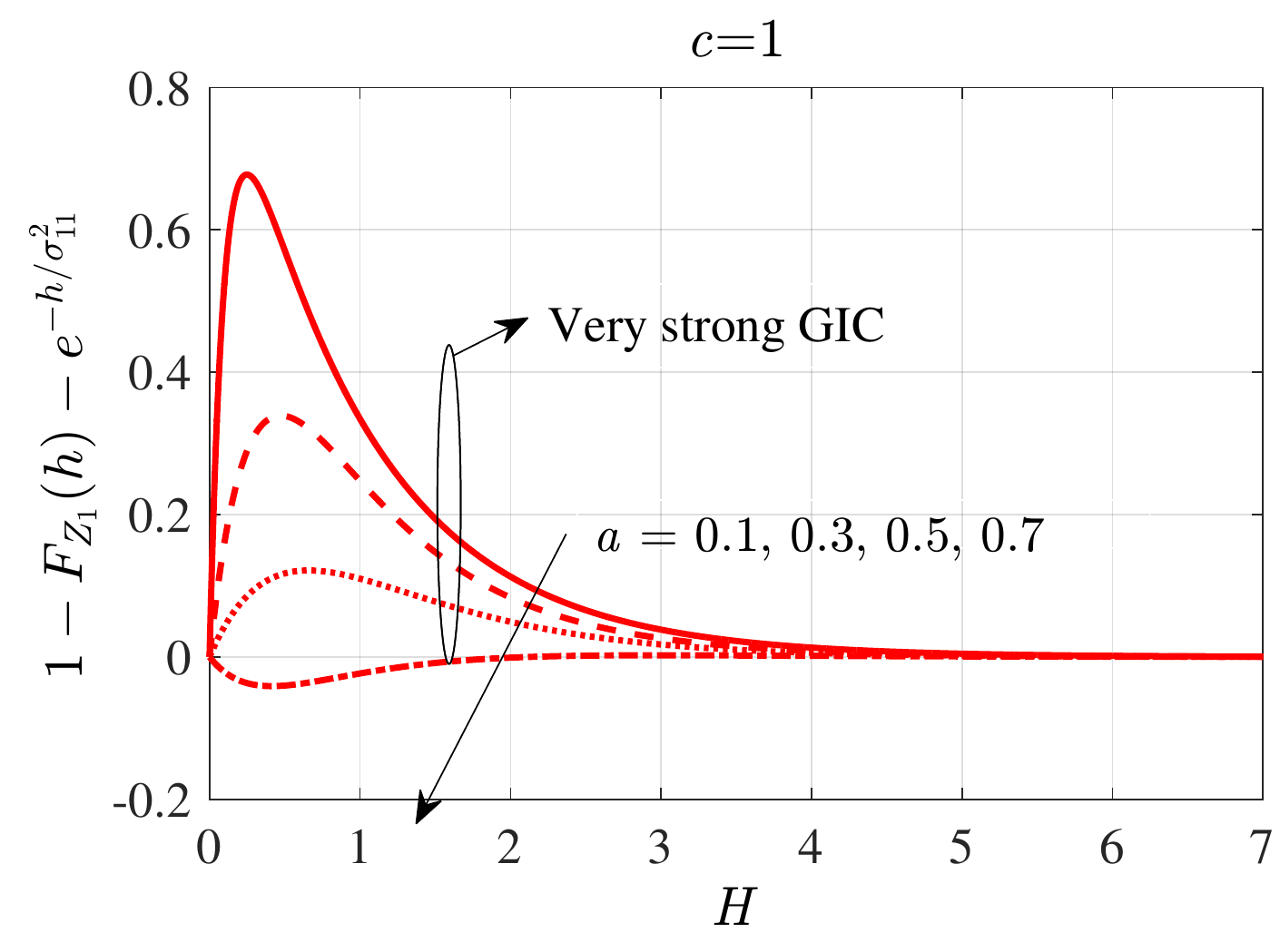, width=0.5\textwidth}
\caption{Identification of the constraints in \eqref{Constraint_very_strong} under different variances of channel gains of the dedicated channels with $c=1$.}
\label{Fig_very_strong_diff_a}
\end{figure}

\begin{figure}[h]
\centering \epsfig{file=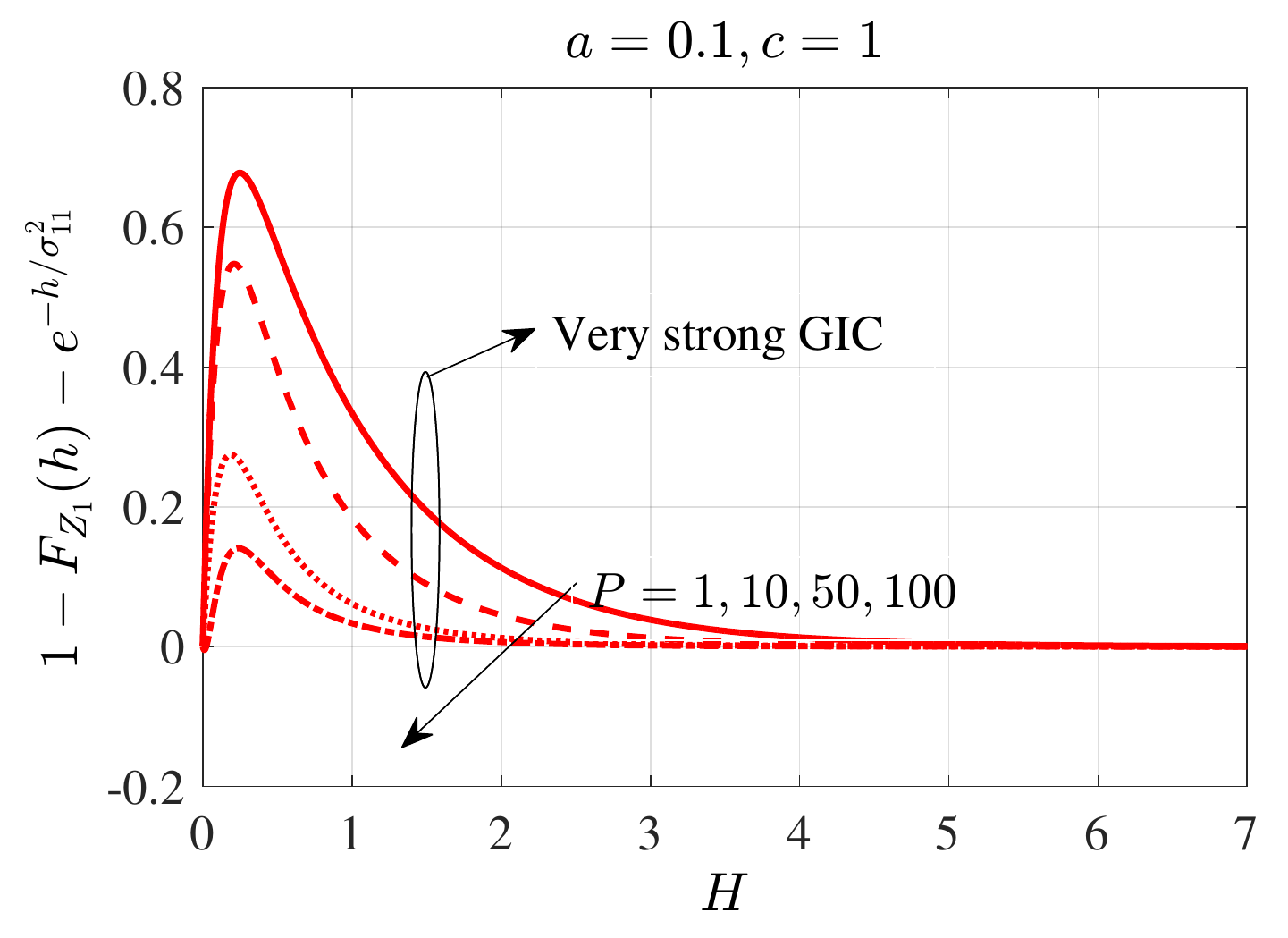, width=0.5\textwidth}
\caption{Identification of the constraints in \eqref{Constraint_very_strong} under different transmit power constraints with $a=0.1$ and $c=1$.}
\label{Fig_very_strong_diff_P}
\end{figure}

\subsection{Fading Gaussian Wiretap Channels with Statistical CSIT}\label{Sec_PHYSEC}

%
%
%
%
For GWTCs with statistical CSIT, compared to \cite{PH_fading_UB}, here we provide a more intuitive and complete derivation for the secrecy capacity under the weak secrecy constraint. The received signals at the legitimate receiver and the eavesdropper are given respectively as:
\begin{align}
Y=\sqrt{H}X+N_Y,\,\,Z=\sqrt{G}X+N_Z,\notag
\end{align}
where $H$ and $G$ are channel gains from Alice to Bob and Eve, respectively, while $N_Y$ and $N_Z$ are independent AWGNs at Bob and Eve, respectively. Without loss of generality, we assume both $N_Y$ and $N_Z$ being with zero means and unit variances. Assume that Alice only knows the statistics of both $H$ and $G$, while Bob knows perfectly the realization of $H$ and Eve knows perfectly both $G$ and $H$. Denote the transmit power constraint by $P_T$.

\begin{Theorem}\normalfont\label{Th_WTC}
The selection of the following set $\mathcal{S}$ leads to degraded GWTCs in $\mathcal{B}$:
\begin{align}
\mathcal{A}&=\{(H,\, G):\, H\geq_{st} G\},\label{Constraint_WTC_degraded}\\
f'(\mathcal{A})&=\{(\tilde{H},\,\tilde{G}):\,\tilde{H} = F_{H}^{-1}(U),\,\, \tilde{G} = F_{G}^{-1}(U),\,\,U\sim\mbox{Unif}(0,1)\},\label{EQ_WTC_f}\\
\mathcal{B}&=\{(\tilde{H},\,\tilde{G})\in f'(\mathcal{A})\},
\end{align}
then the ergodic secrecy capacity with statistical CSIT of both $H$ and $G$ is
\begin{align}\label{EQ_CWTC}
{C}_S(P_T)=\mathds{E}[C(H P_T)-C(G P_T)].
\end{align}
\end{Theorem}
\begin{proof}
The proof is relegated to \ref{APP_pf_WTC}.
\end{proof}

\begin{rmk}\label{rmk_explanation_gaussian_optimal}\normalfont
{Recall that for a degraded GBC with only statistical CSIT, we cannot claim that the capacity region is achievable by a Gaussian input. On the contrary, for a degraded GWTC with only statistical CSIT, we can prove that a Gaussian input is optimal. The difference can be explained in the following. For a GWTC, there is only one message to be transmitted to the single dedicated receiver and the Markov chain $U - X - Y -Z$ reduces to $X - Y - Z$ \cite{csiszar1978broadcast} when it is degraded, i.e., no prefixing is needed. The simplification on solving the optimal channel input distribution from two random variables to only one makes it easier to prove that Gaussian input is optimal. Note that when the GWTC is not degraded, the optimal $P_{U,X}$ is unknown in general. On the other hand, for a two-user degraded GBC there are two messages, both $U$ and $X$ carry messages in the chain $U - X - Y- Z$.

}
\end{rmk}

\begin{rmk}\label{rmk_SKG}\normalfont
The developed framework can also be applied to secret key generation (SKG), using either channel or source model \cite{Bloch_book}. Since the SKG based on channel model can be derived from a conceptual WTC \cite{Maurer_SKG93}, which can be deduced from Theorem \ref{Th_WTC}, we focus on the discussion of SKG based on source model. In the source model, there exists a common random source observed by Alice, Bob, and Eve. In Gaussian Maurer's (satellite) model \cite{Naito_satellite}, the observed signal is the common random source passing through an AWGN channel. If that source is affected by the medium, e.g., the wireless channel, we can form a fading Gaussian Maurer's (satellite) model, which is similar to a fading GBC with an additional eavesdropper. Related discussions can be found in \cite{PHL_SSKG_J}.
\end{rmk}

\section{Channels with Memory}\label{Sec_memory}
In this section, we discuss channels with a specific structure of memory, namely, channels with finite-state \cite{McMillan_FSC}. In particular, we investigate the relationship between the stochastic orders for random processes and the ergodic capacity of a BC with finite state as a representative example. Due to the memory, the concept of degradedness and the same marginal property have to be carefully revisited. Then we discuss the usual stochastic order for random processes.

\subsection{Preliminaries}
\begin{definition}[\normalfont\textbf{Finite state broadcast channels (FSBC) } \cite{Dabora_degraded_finite_state_BC}]\label{Def_FSBC}
The discrete finite-state broadcast channel is defined by the triplet $\{\mathcal{X}\times\mathcal{S},\,p(y,z,s|x,s'),\,\mathcal{Y}\times\mathcal{Z}\times\mathcal{S}\}$, where $X$ is the input symbol, $Y$ and $Z$ are the output symbols, $S'$ and $S$ are the channel states at the end of the previous and the current time instants, respectively. $\mathcal{S},\,\mathcal{X},\,\mathcal{Y}$, and $\mathcal{Z}$ are finite sets. The PMF of the FSBC satisfies
\begin{align}
p(y_i,\,z_i,\,s_i|y^{i-1},\,z^{i-1},\,s^{i-1},\,x^{i},\,s_0)=p(y_i,\,z_i,\,s_i|x_i,\,s_{i-1}),
\end{align}
where $s_0$ is the initial channel state.
\end{definition}

For FSBC, a single letter expression of the condition for degradedness is in general not possible. Therefore we introduce the following definitions of physical and stochastic degradedness for FSBC.

\begin{definition}[\normalfont\textbf{Degradedness for FSBC} \cite{Dabora_degraded_finite_state_BC}]\label{Def_degraded_memory}
An FSBC is called physically degraded if for every $s_0\in\mathcal{S}$ and time instant $i$ its PMF satisfies
\begin{align}
p(y_i|x^i,y^{i-1},z^{i-1},s_0)&=p(y_i|x^i,y^{i-1},s_0),\\
p(z_i|x^i,y^{i},z^{i-1},s_0)&=p(z_i|y^i,z^{i-1},s_0).
\end{align}
The FSBC is called stochastically degraded if there exists a PMF $\tilde{p}(z|y)$ such that for every block length $n$ and initial state $s_0\in\mathcal{S}$ such that
\begin{align}
p(z^n|x^n,s_0)&=\sum_{\mathcal{Y}^n}p(y^n,z^n|x^n,s_0)=\sum_{\mathcal{Y}^n}p(y^n|x^n,s_0)\prod_{i=1}^n \tilde{p}(z_i|y_i).\label{EQ_stochastic_degraded}
\end{align}
\end{definition}

Note that two important properties of the FSBC considered in Definition \ref{Def_degraded_memory} are: 1) The channel output at the weaker receiver does not contain more information than that of the stronger receiver. 2) The stronger output up to the current time instant makes the causally degraded output independent to the channel input up to the current time instant. Note also that Definition \ref{Def_degraded_memory} can be easily specialized to the degradedness of the memoryless case.


In this work, we consider the \textit{indecomposable} FSBC (IFSBC) \cite{Gallager_book}, where the effect of the initial channel state $s_0$ on the state transition probabilities diminishes over time. To apply the concept of stochastic orders, similar to the aforementioned memoryless channels, we need the same marginal property. However, to take the memory effect into account, again, we need to consider a multi-letter version, which can be easily proved by removing the memoryless property in the proof of the same marginal property of the memoryless case, e.g., in  \cite[Theorem 13.9]{Moser_adv}.


\begin{Corollary}\label{lemma_same_marginal_memory}\normalfont
The capacity region of an IFSBC depends only on the conditional marginal distributions $p_{Y^n|X^n,\,S_0}$ and $p_{Z^n|X^n,\,S_0}$ and not on the joint conditional distribution $p_{Y^n,Z^n|X^n,\,S_0}$.
\end{Corollary}

\subsection{Sufficient Conditions for an IFSBC to be Degraded}
In this subsection, we identify the condition that an IFSBC fulfills the usual stochastic order. Again we may invoke the coupling scheme combined with the same marginal property to form an equivalently degraded channel, where for all time instants, fading states to all receivers follow the same trichotomy order. However, when the memory effect is taken into account, the stochastic order introduced in the previous section is not sufficient. On the contrary, we need to resort to the stochastic order for a random process, which is capable of capturing the time structure of memory channels. We focus our discussion on the time homogeneous\footnotemark$\,$  finite-state Markov channels \footnotetext{The transition probability matrix does not change over time.}for two main reasons. First, it is useful for modeling mobile wireless channels when the underlying channel state process is time-invariant and indecomposable \cite{Das_MAC_Markov}, \cite{Wang_markov_fading}, \cite{Goldsmith_PhD_thesis}, \cite{SPM_Markov_fading}. Second, the structure of a Markov chain simplifies the sufficient condition of usual stochastic order for random processes, which increases the tractability and provides insights on the analysis of fading channel with memory.
%



Consider a $k$-th order time-homogeneous Markov process with alphabet $\mathcal{T}$, $|\mathcal{T}|=N$, described by a transition probability matrix $\bm P$, which is fixed over time. Indices of each row and column of $\bm P$ represent the current \textit{super states} $\bm s=(s_{j},s_{j+1},\cdots,s_{j+k-1})$ and the next super state $\bm s'=(s_{j+1},s_{j+2},\cdots,s_{j+k})$, respectively, where $j\in\mathds{N}$. Denote the $j$-th row of $\bm P$ as $\bm p_j$. Then the transition probability from the super state $\bm s$ to $\bm s'$ is expressed as $P_{\bm s\rightarrow \bm s'}\triangleq  \frac{P_{s_{j+k},\bm s}}{P_{\bm s}}=P_{s_{j+k}|\bm s}$. To simplify the expression, we define the following mapping
\begin{align}\label{EQ_g}
g: L \mapsto\bm s,
\end{align}
where $L\in\{1,\cdots,N^k\}$ is an integer indicating the row of $\bm P$ and $\bm s\in \mathcal{T}^k$.


\begin{Theorem}\normalfont\label{Th: st_stochastic_process}
Consider a two-user indecomposable finite-state Markov fading BC where the component Markov fading channels are both with $k$-th order, namely, $\{H_1(m)\}$ and $\{H_2(m)\}$ with fading states arranged in an increasing manner with respect to the state-values, where $m$ is the time index, and with transition matrices $\bm {P}$ and $\bm {Q}$, respectively. Then the BC is degraded if
\begin{align}
\mathcal{A}=&\{(\{H_1(m)\},\{H_2(m)\}):\,\notag\\
&H_1(0)\leq_{st} H_2(0), \label{EQ_Markov_condition1}\\
&[H_1(m)|H_1(m-1)=h_1(m-1),\cdots,H_1(0)=h_1(0)]\leq_{st}\notag\\
   &[H_2(m)|H_2(m-1)=h_2(m-1),\cdots,H_2(0)=h_2(0)],\notag\\
   &\mbox{whenever $h_1(j)\leq h_2(j),\,j=0,\,1,\,\cdots,\,m-1,$ \mbox{if $1\leq m\leq k$, and,}} \label{EQ_Markov_condition2}\\
&\bar{F}_{\bm {p_l}}(n)\leq \bar{F}_{\bm {q_s}}(n),\,\forall\,n,\,\forall \,(l,s)=\left\{l,s|g(l)\leq g(s)\right\},\,\mbox{if } m>k\},\label{EQ_Markov_condition3}
\end{align}
where $\bar{F}_{\bm {p_i}}$ and $\bar{F}_{\bm {q_i}}$ are the CCDFs of the $i$-th state of $\{H_1(m)\}$ and $\{H_2(m)\}$, respectively, $n$ is the index of channel states, and $\leq$ compares two vectors element-wisely.
\end{Theorem}
\begin{proof}
The proof is relegated to \ref{APP_pf_memory}.
\end{proof}
Note that conditions \eqref{EQ_Markov_condition1} and \eqref{EQ_Markov_condition2} should be identified according to initial conditions which are additionally given, but cannot be derived from the transition matrices $\bm P$ and $\bm Q$.

The first order Markov fading channels can be further simplified from Theorem \ref{Th: st_stochastic_process} shown as follows.

\begin{Corollary}\normalfont\label{coro_Markov_condition_1st_order}
Consider two first-order Markov fading channels $\{H_1(m)\}$ and $\{H_2(m)\}$ with fading states arranged in an increasing manner with respect to the state-values. A two-user finite-state Markov fading BC is degraded if
\begin{align}\label{EQ_Markov_condition_1st_order}
&H_1(0)\leq_{st} H_2(0), \mbox{ and}\\
&\bar{F}_{\bm {p}_l}(n)\leq \bar{F}_{\bm {q}_s}(n),\,\forall\,n,\,\forall\,l,\,s,\,\,s.t.\,\,\,l\leq s.\label{EQ_1st_order_transition_constraint}
\end{align}
\end{Corollary}

\begin{rmk}\normalfont
For an IFSBC which is degraded, the capacity region can be described by \cite{Dabora_degraded_finite_state_BC}
\begin{align}
\mathcal{C}=\underset{n\rightarrow\infty}\lim co\bigcup_{q_n\in\mathcal{Q}_n}\Bigg\{(R_1,R_2):\,R_1\geq 0,\,R_2\geq\,0,\,R_1\leq\frac{1}{n}I(X^n;Y^n|U^n,s_0')_{q_n},\,
R_2\leq\frac{1}{n}I(U^n;Z^n|s_0'')_{q_n}\Bigg\},\label{EQ_IFSBC_capacity}
\end{align}
where $\mathcal{Q}_n$ is the set of all joint distributions $p(u^n,x^n)$, $s_0'$ and $s_0''$ are arbitrary initial states. For WTCs with finite state memory and statistical CSIT, we can also apply the above discussion to identify the degradedness. From \cite{Yogesh} we know that there is no need to optimize the channel prefixing $p(x^n|u^n)$, if the WTC is degraded. Under the assumption of statistical CSIT, if the conditions in Theorem \ref{Th: st_stochastic_process} are valid, we know that the secrecy capacity of a fading WTC with finite state can be described by \cite[Corollary 1]{Bloch_resolvability}.
\end{rmk}

%


In the following we provide two examples to show an application of our results to fading channels with memory.\\
\begin{Ex} Consider two three-state first order Markov chains. Given any $H_2(0)\leq_{st} H_1(0)$, the following transition probability matrices of $H_1$ and $H_2$, respectively, satisfy Corollary \ref{coro_Markov_condition_1st_order} and form a degraded BC
\begin{align}\label{EQ_1st_order_Markov_Transition_matrix}
\bm P=\left(
    \begin{array}{ccc}
      \frac{1}{2} & \frac{1}{4} & \frac{1}{4} \\
      \frac{3}{4} & \frac{1}{8} & \frac{1}{8} \\
      \frac{5}{8} & \frac{1}{4} & \frac{1}{8} \\
    \end{array}
  \right),\,\,\bm Q=\left(
    \begin{array}{ccc}
      \frac{1}{4} & \frac{3}{8} & \frac{3}{8} \\
      \frac{1}{8} & \frac{2}{8} & \frac{5}{8} \\
      \frac{1}{2} & \frac{1}{8} & \frac{3}{8} \\
    \end{array}
  \right).
\end{align}
%

To identify the degradedness, we resort to Corollary \ref{coro_Markov_condition_1st_order}, and verify the condition for the corresponding CCDF matrices as
\begin{align}\label{EQ_CCDF_P_Q}
\bar{F}_{\bm P}=\left(
    \begin{array}{ccc}
      \frac{1}{2} & \frac{1}{4} & 0 \\
      \frac{1}{4} & \frac{1}{8} & 0 \\
      \frac{3}{8} & \frac{1}{8} & 0 \\
    \end{array}
  \right),\,\,\bar{F}_{\bm Q}=\left(
    \begin{array}{ccc}
      \frac{3}{4} & \frac{3}{8} & 0 \\
      \frac{7}{8} & \frac{5}{8} & 0 \\
      \frac{1}{2} & \frac{3}{8} & 0 \\
    \end{array}
  \right).
\end{align}
After comparing the pairs of CCDFs: $(\bar{F}_{\bm {p}_1},\bar{F}_{\bm {q}_1})$, $(\bar{F}_{\bm {p}_1},\bar{F}_{\bm {q}_2})$, $(\bar{F}_{\bm {p}_1},\bar{F}_{\bm {q}_3})$, $(\bar{F}_{\bm {p}_2},\bar{F}_{\bm {q}_2})$, $(\bar{F}_{\bm {p}_2},\bar{F}_{\bm {q}_3})$, $(\bar{F}_{\bm {p}_3},\bar{F}_{\bm {q}_3})$, it is clear that \eqref{EQ_1st_order_transition_constraint} is valid. Therefore, it is degraded.\\
\end{Ex}

\begin{Ex}

Consider two binary-valued second order Markov chains. The general transition matrix can be expressed as \cite{Pimentel_markov_fading}
\begin{align}
  \bordermatrix{~    & 00 & 01 & 10 & 11 \cr
                  00 & \frac{P_{000}}{P_{00}} & \frac{P_{001}}{P_{00}} & 0 & 0 \cr
                  01 &  0 & 0& \frac{P_{010}}{P_{01}} & \frac{P_{011}}{P_{01}}  \cr
                  10 & \frac{P_{100}}{P_{10}} & \frac{P_{101}}{P_{10}} & 0 & 0 \cr
                  11 & 0 & 0 & \frac{P_{110}}{P_{11}} & \frac{P_{111}}{P_{11}} \cr},
\end{align}
where row and column indices show the current and next super states, respectively.

Consider the following transition matrices
\begin{align}
\bm P=\left(
    \begin{array}{cccc}
      \frac{1}{2} & \frac{1}{2} & 0 & 0 \\
      0 & 0 & \frac{1}{3} & \frac{2}{3} \\
      \frac{1}{4} &       \frac{3}{4} & 0 & 0 \\
      0& 0& \frac{1}{5} & \frac{4}{5} \\
    \end{array}
  \right),\,\,\bm Q=\left(
    \begin{array}{cccc}
     \frac{1}{3} & \frac{2}{3} & 0 & 0 \\
      0 & 0 & \frac{1}{4} & \frac{3}{4} \\
      \frac{1}{5} &       \frac{4}{5} & 0 & 0 \\
      0& 0& \frac{1}{6} & \frac{5}{6} \\
    \end{array}
  \right).\label{EX_2nd_order_Markov}
\end{align}

The corresponding CCDF matrices can be easily derived as
\begin{align}\label{EQ_CCDF_P_Q_2nd_order}
\bar{F}_{\bm P}=\left(
    \begin{array}{cccc}
      \frac{1}{2} & 0 & 0 & 0 \\
      1 & 1 & \frac{2}{3} & 0 \\
      \frac{3}{4} &       0 & 0 & 0 \\
      1& 1& \frac{4}{5} & 0 \\
    \end{array}
  \right),\,\,\bar{F}_{\bm Q}=\left(
    \begin{array}{cccc}
     \frac{2}{3} & 0 & 0 & 0 \\
      1 & 1 & \frac{3}{4} & 0 \\
      \frac{4}{5} &       0 & 0 & 0 \\
      1& 1& \frac{5}{6} & 0 \\
    \end{array}
  \right).
\end{align}
By comparing the pairs of CCDFs: $(\bar{F}_{\bm {p}_1},\bar{F}_{\bm {q}_1})$, $(\bar{F}_{\bm {p}_1},\bar{F}_{\bm {q}_2})$, $(\bar{F}_{\bm {p}_1},\bar{F}_{\bm {q}_3})$, $(\bar{F}_{\bm {p}_1},\bar{F}_{\bm {q}_4})$, $(\bar{F}_{\bm {p}_2},\bar{F}_{\bm {q}_2})$, $(\bar{F}_{\bm {p}_2},\bar{F}_{\bm {q}_4})$, $(\bar{F}_{\bm {p}_3},\bar{F}_{\bm {q}_3})$, $(\bar{F}_{\bm {p}_3},\bar{F}_{\bm {q}_4})$, $(\bar{F}_{\bm {p}_4},\bar{F}_{\bm {q}_4})$, we can verify \eqref{EQ_Markov_condition3}. In addition, given any initial conditions satisfying \eqref{EQ_Markov_condition1},
\begin{align}
[H_2(1)|H_2(0)=h_2(0)]&\leq_{st} [H_1(1)|H_1(0)=h_1(0)],\label{EQ_2nd_order_inital_condition2}\\
[H_2(2)|H_2(1)=h_2(1),\,H_2(0)=h_2(0)]&\leq_{st}[H_1(2)|H_1(1)=h_1(1),\,H_1(0)=h_1(0)],\label{EQ_2nd_order_inital_condition1}
\end{align}
we can identify that it is degraded.
\end{Ex}

\section{Conclusion}\label{Sec_conclusion}
In this paper, we investigated the ergodic capacity of several fading memoryless Gaussian multiuser channels, when only the statistics of the channel state information are known at the transmitter. We first classify the fading channels through their joint probability distributions by
which we are able to obtain the capacity results. Schemes from the maximal coupling, coupling, and copulas are derived and the interrelation is characterized. Based on the classification, we derive sufficient conditions to obtain the capacity regions. Results include Gaussian interference channels, Gaussian broadcast channels and Gaussian wiretap channels. Extension of the framework to channels with finite-state memory is also considered, wherein the Markov fading channel is discussed as a special case. Practical examples illustrate the successful applications of the derived results.

\renewcommand{\thesection}{Appendix I}
\section{Proof of Theorem \ref{Th_all}}\label{APP_all}
We divide the proof into three cases, corresponding to $f$, $f'$, and $f''$, respectively.

\textit{Validation of $f$:} We first prove that by combining maximal coupling with the choice of $\mathcal{A}$ in \eqref{EQ_coupling_construction}, we can solve $\textbf{Problem 1}$. We then show that \eqref{EQ_max_coupling1} and \eqref{EQ_max_coupling2}  indeed achieve the maximal coupling.

Recall that in the proof we have two goals:
\begin{itemize}
\item[G1.] To align realizations of those channel gain pairs both belonging to $\min\{f_{H_1}(h), f_{H_2}(h)\},\,\forall h$, in the sense that $\tilde{h}_1=\tilde{h}_2$;
\item[G2.] The remaining realizations of channel pairs follow a unique trichotomy order.
\end{itemize}

For this purpose, we introduce a result of the maximal coupling.
\begin{Prop}\cite[Proposition 2.5]{Ross_2nd_course}\label{Prop_Ross_Max_coupling}
Suppose $X$ and $Y$ are random variables with respective piecewise continuous density functions $f_X$ and $f_Y$. The maximal coupling $(\tilde{X},\tilde{Y})$ for $(X,\,Y)$ results in
\begin{align}\label{EQ_max_coupling_result}
{\textnormal{ Pr(}}\tilde{X}=\tilde{Y})=\int_{-\infty}^{\infty}\min(f_X(x),\,f_Y(x))dx.
\end{align}
\end{Prop}

Now we show that \eqref{EQ_max_coupling1} and \eqref{EQ_max_coupling2} can achieve the maximal coupling. To be self-contained, we restate the important steps of the proof of \cite[Proposition 2.5]{Ross_2nd_course} in the following. Define $\mathcal{H}=\{h:f_{H_1}(h)<f_{H_2}(h)\}$. Any coupling $(\tilde{H}_1,\,\tilde{H}_2)$ of $(H_1,\,H_2)$ should satisfy
\begin{align}
\mbox{Pr(}\tilde{H}_1=\tilde{H}_2)&=\mbox{Pr(}\tilde{H}_1=\tilde{H}_2,\,\tilde{H}_1\in\mathcal{H})
+\mbox{Pr(}\tilde{H}_1=\tilde{H}_2,\,\tilde{H}_2\in\mathcal{H}^c)\notag\\
&\leq \mbox{Pr(}\tilde{H}_1\in\mathcal{H})+\mbox{Pr(}\tilde{H}_2\in\mathcal{H}^c)\notag\\
&\overset{(a)}= \mbox{Pr(}{H}_1\in\mathcal{H})+\mbox{Pr(}{H}_2\in\mathcal{H}^c)\notag\\
&=\int_{\mathcal{H}}f_{H_1}(h)dh+\int_{\mathcal{H}^c}f_{H_2}(h)dh\notag\\
&\overset{(b)}=\int_{-\infty}^{\infty}\min(f_{H_1}(h),\,f_{H_2}(h))dh\notag\\
&\overset{(c)}=p,\label{EQ_max_coupling_forward}
\end{align}
where (a) is by Definition \ref{Def_coupling}, (b) is by the definition of $\mathcal{H}$ and (c) is by the definition of $p$ in Theorem \ref{Th_all}.

By invoking Proposition \ref{Prop_Ross_Max_coupling} we know that, due to maximal coupling, \textit{all} channel pairs belonging to  $\min\{f_{H_1}(h), f_{H_2}(h)\}$ can be aligned such that $\tilde{h}_1=\tilde{h}_2$, i.e., whose probability is as \eqref{EQ_max_coupling_result}, which fulfills G1. To achieve G2, we require that the supports of $\tilde{h}_1$ and $\tilde{h}_2$ belonging to $f_{H_1}(h)-f_{min}(h)$ and $f_{H_2}(h)-f_{min}(h)$, respectively, do not intersect. Otherwise, there is still an ambiguity in the orders again. Now we select the feasible set $\mathcal{A}$. To ensure that the trichotomy order between the realizations of $\tilde{H}_1$ and $\tilde{H}_2$ in \eqref{EQ_max_coupling1} is fixed, a sufficient condition is to enforce that $\tilde{f}_{H_1}$ and $\tilde{f}_{H_2}$ from \eqref{EQ_max_coupling1} do not intersect:
\begin{align}
\sup \mbox{ supp}( \tilde{f}_{H_1} )<\inf \mbox{ supp}( \tilde{f}_{H_2} ).
\end{align}

To proceed, we prove the following result:
\begin{align}\label{EQ_iff_cond_usual_stoch_order}
\sup \mbox{ supp}( \tilde{f}_{H_1} )<\inf \mbox{ supp}( \tilde{f}_{H_2} )\mbox{   if and only if   } H_1\leq_{st} H_2.
\end{align}
We first prove the "only if" direction. We have the following relations:
\begin{align}
H_1\leq_{st} H_2 \mbox{ if and only if } & \int_0^x f_{H_1}(h)dh\geq \int_0^x f_{H_2}(h)dh, \,\forall x,\label{EQ_int_fk-fmin0}\\
\mbox{ if and only if }  & \int_0^x \tilde{f}_{H_1}(h)dh\geq \int_0^x \tilde{f}_{H_2}(h)dh,\,\forall x,\label{EQ_int_fk-fmin}
\end{align}
where \eqref{EQ_int_fk-fmin0} is by definition of usual stochastic order; \eqref{EQ_int_fk-fmin} is by subtracting both sides on the RHS of (a) by $\int_0^x f_{min}(h)dh$. It is clear that, if $\sup \mbox{ supp}( f_{H_1}-f_{min} )<\inf \mbox{ supp}( f_{H_2} -f_{min})$, then \eqref{EQ_int_fk-fmin} is valid and hence $H_1\leq_{st}H_2$.

Now we prove the "if" direction by contradiction. We first rewrite \eqref{EQ_max_coupling1} as
\begin{align}
f_{H_k}(h)&=(1-p)\tilde{f}_{H_k}(h)+f_{min}(h),\,\,k=1,\,2,\,\forall h,\notag
\end{align}
or,
\begin{align}\label{EQ_F_tild_F}
F_{H_k}(h)&=(1-p)\tilde{F}_{H_k}(h)+F_{min}(h),\,\,k=1,\,2,\,\forall h.
\end{align}
By subtracting $F_{H_1}$ from $F_{H_2}$ according to \eqref{EQ_F_tild_F}, we have
\begin{align}\label{EQ_diff_F1_F2}
0\geq F_{H_2}(h)-F_{H_1}(h)=(1-p) \left(\tilde{F}_{H_2}(h)-\tilde{F}_{H_1}(h)\right),\,\,\forall h,
\end{align}
where the inequality is due to $H_1\leq_{st}H_2$. Since $1-p\geq 0$ and by definition of usual stochastic order, we can equivalently express \eqref{EQ_diff_F1_F2} by
\begin{align}\label{EQ_sto_ord_contradiction}
\tilde{H}_1\leq_{st}\tilde{H}_2.
\end{align}
To show the contradiction, we assume that
\begin{align}\label{EQ_counter_proposition}
\sup \mbox{ supp}( \tilde{f}_{H_1})\geq \inf \mbox{ supp}( \tilde{f}_{H_2}).
\end{align}
Since the intersection of $\tilde{f}_{H_1}$ and $\tilde{f}_{H_2}$ is null from \eqref{EQ_max_coupling1}, there are only two possibilities to attain \eqref{EQ_counter_proposition}:
\begin{enumerate}
\item $\tilde{f}_{H_2}$ has at least one disconnected part $\tilde{f}_{H_2}'$, such that $\inf\mbox{ supp}(\tilde{f}_{H_2}')>\sup\mbox{ supp}(\tilde{f}_{H_1})$. However, such $\tilde{f}_{H_2}$ results in cross points between $\tilde{F}_{H_1}$ and $\tilde{F}_{H_2}$ within the open interval (0, 1) of the range of $\tilde{F}_{H_1}$ and $\tilde{F}_{H_2}$, which violates \eqref{EQ_sto_ord_contradiction}  and then also violates $H_1\leq_{st}H_2$.
\item $\tilde{f}_{H_2}$ is connected. Then to attain \eqref{EQ_counter_proposition}, we can only have $\tilde{f}_{H_1}(h)>\tilde{f}_{H_2}(h),\,\forall\,h$, which violates \eqref{EQ_sto_ord_contradiction}  and then also violates $H_1\leq_{st}H_2$.
\end{enumerate}
Since both cases violate the assumption, we complete the proof of \eqref{EQ_iff_cond_usual_stoch_order}.

From \eqref{EQ_iff_cond_usual_stoch_order}, it is clear that the set $\mathcal{A}$ which can preserve the trichotomy order of the realizations generated according to \eqref{EQ_max_coupling1}, can be ensured by the usual stochastic order. Then by combining \eqref{EQ_coupling_construction} and \eqref{EQ_max_coupling1} we know that with probability $1-p$, we have $\tilde{h}_{1}<\tilde{h}_{2}$. As the result, by the maximal coupling, we can construct equivalent channels where there are only two relations between the fading channels realizations $\tilde{h}_1$ and $\tilde{h}_2$: 1) $\tilde{h}_1=\tilde{h}_2$, with probability $p$, and, 2) $\tilde{h}_{1}<\tilde{h}_{2} $, with probability $1-p$.

The selection of $(\tilde{H}_1,\,\tilde{H}_2)$ in \eqref{EQ_max_coupling1} and \eqref{EQ_max_coupling2} can be verified as a coupling as follows. Assume the PDF of $\tilde{H}_k$ is switched between \eqref{EQ_max_coupling1} and \eqref{EQ_max_coupling2}, controlled by $V\sim\mbox{Bern}(p)$. Then, we have:
\begin{align}
\mbox{Pr(}\tilde{H}_k\leq h')\overset{(a)}=& p\cdot\mbox{Pr(}\tilde{H}_k\leq h'|V=1)+(1-p)\cdot\mbox{Pr(}\tilde{H}_k\leq h'|V=0)\notag\\
=&p\int_{\infty}^{h'} \frac{\min(f_{H_1}(h),\,f_{H_2}(h))}{p} dh+(1-p)\int_{\infty}^{h'} \frac{f_{H_k}(h)-\min(f_{H_1}(h),\,f_{H_2}(h))}{1-p} dh\notag\\
=&\int_{\infty}^{h'}f_{H_k}(h)dh=\mbox{Pr(}{H}_k\leq h'),\,\mbox{ for } k=1,\,2,\label{EQ_check_coupling}
\end{align}
where (a) is by law of total probability. Hence, it is clear that \eqref{EQ_check_coupling} fulfills the definition of coupling in Definition \ref{Def_coupling}. On the other hand, it is clear that $\mbox{Pr(}\tilde{H}_1=\tilde{H}_2)\geq \mbox{Pr(}\tilde{H}_1=\tilde{H}_2,\,V=1)=p$. Therefore, from \eqref{EQ_max_coupling_forward} and \eqref{EQ_check_coupling}, we know that \eqref{EQ_max_coupling1} and \eqref{EQ_max_coupling2} can achieve the maximal coupling, which completes the proof of the first case.\\
\textit{Validation of $f'$:}
The proof of the coupling theorem \cite[Ch. 2]{Ross_2nd_course} provides us a constructive way to find $f'$, which is restated as follows for a self contained proof. If $H_1\leq_{st} H_2$, and if the generalized inverses $F_{H_1}^{-1}$ and $F_{H_2}^{-1}$ exist, where the generalized inverse of $F$ is defined by $F^{-1}(u)=\inf
\{x\in\mathds{R}:\,F(x)\geq u,\,u\in[0,\,1]\}$, then the equivalent channels $\tilde{H}_1$ and $\tilde{H}_2$ can be constructed by $\tilde{H}_1=F_{H_1}^{-1}(U)$ and $\tilde{H}_2=F_{H_2}^{-1}(U)$, respectively, where $U\sim \mbox{Unif}(0,1)$. This is because
\begin{align}
\mbox{Pr(}\tilde{H}_1\leq h)=\mbox{Pr(}F_{H_1}^{-1}(U)\leq h)=\mbox{Pr(}U\leq F_{H_1}(h))=F_{H_1}(h),
\end{align}
i.e., $\tilde{H}_1$ has the same CDF as ${H_1}$. Similarly, $\tilde{H}_2$ has the same CDF as ${H_2}$. Since $H_1\leq_{st} H_2$, from Definition \ref{shaked_stochastic_order} we know that $F_{H_2}(h)\leq F_{H_1}(h)$, for all $h$. Then it is clear that $F_{H_1}^{-1}(u)\leq F_{H_2}^{-1}(u),\,\mbox{ for all } u\in [0,1]$, such that $\mbox{Pr(}\tilde{H}_1\leq \tilde{H}_2)=1$. Therefore, we attain \eqref{EQ_construct_f}, which completes the proof of the second case.

\textit{Validation of $f''$:} We first derive a joint distribution from the Fr\'{e}chet-Hoeffding-upper bound for the survival copulas. Then we show the validity of $f''$ by proving that $f''$ is equivalent to $f'$. From Fr\'{e}chet-Hoeffding bounds \cite[Sec. 2.5]{Nelson_copulas} we know that a joint CDF ${F}_{H_1H_2}(h_1,h_2)$ can be upper and lower bounded by the marginals $F_{H_1}(h_1)$ and $F_{H_2}(h_2)$ as follows:
\begin{align}\label{EQ_FH_bounds}
\max\{F_{H_1}(h_1)+F_{H_2}(h_2)-1,\,0\}\leq {F}_{H_1H_2}(h_1,h_2)\leq \min\{F_{H_1}(h_1),\,F_{H_2}(h_2)\}.
\end{align}
On the other hand, by definition of the joint CCDF $\bar{F}_{H_1H_2}(h_1,h_2)$ and the joint CDF ${F}_{H_1H_2}(h_1,h_2)$, we can easily see:
\begin{align}\label{EQ_CDF_CCDF}
\bar{F}_{H_1H_2}(h_1,h_2) = 1 - F_{H_1}(h_1) - F_{H_2}(h_2) + {F}_{H_1H_2}(h_1,h_2).
\end{align}
By selecting the upper bound in \eqref{EQ_FH_bounds} as the desired ${F}_{H_1H_2}(h_1,h_2)$ and substitute it into \eqref{EQ_CDF_CCDF}, we can get:
\begin{align}
\bar{F}_{H_1H_2}(h_1,h_2)&= 1 - F_{H_1}(h_1) - F_{H_2}(h_2) + \min\{F_{H_1}(h_1),\,F_{H_2}(h_2)\}\notag\\
&= 1 + \min\{-F_{H_1}(h_1),\,-F_{H_2}(h_2)\}\notag\\
&\overset{(a)}= \min\{\bar{F}_{H_1}(h_1),\,\bar{F}_{H_2}(h_2)\},\label{EQ_joint_CCDF_copula}
\end{align}
where (a) is by the definition of CCDF,

Now we show that $f'$ is equivalent to $f''$. By the definition of joint CCDF, it is clear that the marginal distributions are
unchanged by the construction of $\bar{F}_{\tilde{H}_1,\,\tilde{H}_2}$ in \eqref{EQ_new_joint_CCDF}, i.e.,
$\bar{F}_{\tilde{H}_1}(\tilde{h}_1)=\bar{F}_{\tilde{H}_1,\tilde{H}_2}(\tilde{h}_1,\,0)=\bar{F}_{H_1}(\tilde{h}_1)$
and
$\bar{F}_{\tilde{H}_2}(\tilde{h}_2)=\bar{F}_{\tilde{H}_1,\tilde{H}_2}(0,\,\tilde{h}_2)=\bar{F}_{{H}_2}(\tilde{h}_2)$. Note that channel gains are non-negative, so we substitute 0 into $\bar{F}_{\tilde{H}_1,\,\tilde{H}_2}(\tilde{h}_1,\,\tilde{h}_2)$ to marginalize it. With the selection $\mathcal{A}=\{(H_1,\,H_2): H_1\leq_{st}H_2\}$, we can prove that $\tilde{h}_1\leq \tilde{h}_2,\,\forall (\tilde{H}_1,\,\tilde{H}_2)\in \mathcal{B}= f''(\mathcal{A})$ by showing that $f''$ is equivalent to $f'$. The equivalence can be seen by showing that $(\tilde{H}_1,\tilde{H}_2)$ generated from $f'$ and $f''$ have the same joint CCDF as follows:
\begin{align}\label{EQ_f_prime_eq_f_2prime}
\mbox{Pr(}F_{H_1}^{-1}(U )\leq \tilde{h}_1,\,F_{H_2}^{-1}(U )
\leq \tilde{h}_2 )&=\mbox{Pr(}U\leq F_{H_1}(\tilde{h}_1 ),\,U \leq  F_{H_2}(\tilde{h}_2))\notag\\
& = \mbox{Pr(}U \leq  \min\{F_{H_1}(\tilde{h}_1 ),\,F_{H_2}(\tilde{h}_2 )\})\notag\\
&\overset{(a)}= \min\{F_{H_1}(\tilde{h}_1 ),\,F_{H_2}(\tilde{h}_2 )\},
\end{align}
where (a) is due to the assumption that $U\sim \mbox{Unif}(0,1)$. Then from the RHS of the first equality in \eqref{EQ_joint_CCDF_copula}, we know that the joint CCDF of $\tilde{H}_1$ and $\tilde{H}_2$ incurred from $f'$ is the same as that from $f''$, which completes the proof. \QEDA

\renewcommand{\thesection}{Appendix II}
\section{Proof of Remark \ref{RMK_copula}}\label{APP_copula}
In the following we directly verify that $C_0(F_{H_1}(\tilde{h}_1),\,F_{H_2}(\tilde{h}_2))=\min\{F_{H_1}(\tilde{h}_1 ),\,F_{H_2}(\tilde{h}_2 )\}$ fulfills \eqref{EQ_2.2.a}, \eqref{EQ_2.2.b}, and \eqref{EQ_2.2.3}, respectively. It can be easily seen that $C_0(F_{H_1}(\tilde{h}_1),\,0)=C_0(0,\,F_{H_2}(\tilde{h}_2))=0,$ then \eqref{EQ_2.2.a} is fulfilled. We can also easily see that
\begin{align}
\min\{\bar{F}_{H_1}(\tilde{h}_1),1\}=\bar{F}_{H_1}(\tilde{h}_1)\mbox{ and }\min\{1,\bar{F}_{H_2}(\tilde{h}_2)\}=\bar{F}_{H_2}(\tilde{h}_2),
\end{align}
then \eqref{EQ_2.2.b} is fulfilled. To check \eqref{EQ_2.2.3}, we first define $u_j\triangleq \bar{F}_{H_{1}}(\tilde{h}_{1,j})$ and $v_j\triangleq \bar{F}_{H_{2}}(\tilde{h}_{2,j})$, $j=1,\,2$, where the subscript $j$ indicates the different realizations of $H_1$ and $H_2$. The condition $\{u_1\leq u_2$ and $v_1\leq v_2\}$ is composed by the following cases: 1) $v_1\leq v_2\leq u_1 \leq u_2$, 2) $v_1\leq u_1\leq v_2 \leq u_2$, 3) $v_1\leq u_1\leq u_2 \leq v_2$, 4) $u_1\leq v_1\leq v_2 \leq u_2$, 5) $u_1\leq v_1\leq u_2 \leq v_2$, 6) $u_1\leq u_2 \leq v_1 \leq v_2$. We can further merge the above cases into the following 4 classes:\\
Class 1 (Cases 1 and 2):  $u_2\geq v_2,\,u_1\geq v_1$: the LHS of \eqref{EQ_2.2.3} can be expressed as
\begin{align}\label{EQ_2increasing1}
v_2-v_1-\min(u_1,v_2)+v_1=v_2-\min(u_1,v_2)\geq 0.
\end{align}
Class 2 (Case 4): the LHS of \eqref{EQ_2.2.3} can be expressed as
\begin{align}\label{EQ_2increasing2}
v_2-v_1-u_1+u_1=v_2-v_1\geq 0.
\end{align}
Class 3 (Case 3): the LHS of \eqref{EQ_2.2.3} can be expressed as
\begin{align}\label{EQ_2increasing3}
u_2-v_1-u_1+v_1=u_2-u_1\geq 0.
\end{align}
Class 4 (Cases 5 and 6):  $u_2\leq v_2,\,u_1\leq v_1$: the LHS of \eqref{EQ_2.2.3} can be expressed as
\begin{align}\label{EQ_2increasing4}
u_2-\min(u_2,v_1)-u_1+u_1=u_2-\min(u_2,v_1)\geq 0.
\end{align}
Therefore, from \eqref{EQ_2increasing1}, \eqref{EQ_2increasing2}, \eqref{EQ_2increasing3}, and \eqref{EQ_2increasing4}, the selection of \eqref{EQ_new_joint_CCDF} is a copula, which completes the proof. \QEDA

\renewcommand{\thesection}{Appendix III}
\section{Proof of Theorem \ref{Th_cond_strong_IC}}\label{APP_Th_cond_strong_IC}
We first extend the capacity result of a \textit{uniformly strong IC} (US IC) \cite{Sankar_ergodic_GIC} with perfect CSIT, in which at each time instant the realizations of channel gains satisfy $h_{21}\geq h_{11}$ and $h_{12}\geq h_{22}$, to the case with statistical CSIT, which has not been reported in the literature to the best of our knowledge. Then we generalize the US IC to \eqref{EQ_strong_IC_constraint}. By doing so, we can smoothly connect the stochastic orders with the capacity region. To prove the ergodic capacity region, we extend the proof in \cite[Theorem 3]{Sankar_ergodic_GIC} with proper modifications to fit our assumptions. For the achievable scheme, it is clear that allowing each receiver to decode both messages from the two transmitters provides an inner bound, i.e., \eqref{EQ_CMAC}, of the capacity region.
We now establish a matching outer bound, by showing that \eqref{EQ_CMAC_def} is an outer bound\footnotetext{The single user capacity outer bounds of $R_1$ and $R_2$ can be easily derived. Therefore, here we only focus on the sum capacity outer bound.} of the capacity region of the considered model, where a genie bestows the information of the interference to only one of the receivers, e.g., the second receiver, which is equivalent to setting $h_{21}=0$. By this genie aided channel, we aim to prove\footnotemark.
\begin{align}\label{EQ_genie_UB_sum_rate_final1}
R_1+R_2\leq \mathds{E}\left[C\left(H_{11}P_1+H_{12}P_2\right)\right].
\end{align}

From Fano's inequality we know that the sum rate must satisfy
\begin{align}
&(R_1+R_2-\epsilon)\notag\\
\overset{(a)}\leq& I(X_1^n;Y_1^n|\tilde{\bm H}^n)+I(X_2^n;Y_2^n|X_1^n,\tilde{\bm H}^n)\notag\\
\overset{(b)}=& \mathds{E}[I(X_1^n; \tilde{h}_{11}^n X_1^n+ \tilde{h}_{12}^n X_2^n+N_1^n |\tilde{\bm H}^n=\tilde{\bm h}^n)+I(X_2^n; \tilde{h}_{22}^n X_2^n+N_2^n |\tilde{\bm H}^n=\tilde{\bm h}^n)]\notag\\
=& \mathds{E}[h( \tilde{h}_{11}^n X_1^n+ \tilde{h}_{12}^n X_2^n+N_1^n |\tilde{\bm H}^n=\tilde{\bm h}^n)-h( \tilde{h}_{12}^n X_2^n+N_1^n |\tilde{\bm H}^n=\tilde{\bm h}^n)+h( \tilde{h}_{22}^n X_2^n+N_2^n |\tilde{\bm H}^n=\tilde{\bm h}^n)-h(N_2^n) ]\notag\\
\overset{(c)}\leq& \mathds{E}\left[\sum\limits_{k=1}^n \left(h( \tilde{h}_{11,k} X_{1,k}\!+\! \tilde{h}_{12,k} X_{2,k}\!+\!N_{1,k}|\tilde{\bm H}=\tilde{\bm h})\!-\!h(N_{2,k})\right)\!+\!h( \tilde{h}_{22}^n X_2^n+N_2^n |\tilde{\bm H}=\tilde{\bm h}) \!-\!h( \tilde{h}_{12}^n X_2^n\!+\!N_1^n |\tilde{\bm H}=\tilde{\bm h})\right],\label{EQ_genie_UB_sum_rate_1}
\end{align}
where on the RHS of (a), the condition of the second term is due to the genie and we define $\tilde{\bm H}^n\triangleq[\tilde{H}_{11}^n,\,\tilde{H}_{12}^n,\,\tilde{H}_{22}^n]$; in (b) the expectation is over $\tilde{\bm H}^n$. To simplify the notation, we omit the subscript of $\tilde{\bm H}^n$ in expectation; (c) is by applying the chain rule of entropy and conditioning reduces entropy and i.i.d. property to the first and the fourth terms on the RHS of the second equality, respectively. Since the last term on the RHS of \eqref{EQ_genie_UB_sum_rate_1} has a sign change and also it is not as simple as the term $h(N_2^n)$, we concentrate on the single letterization of $h( \tilde{H}_{22}^n X_2^n+N_2^n|\tilde{\bm H}=\tilde{\bm h} ) -h( \tilde{H}_{12}^n X_2^n+N_1^n |\tilde{\bm H}=\tilde{\bm h})$. To proceed, we exploit the property\footnote{Note that without this property, we may not be able to rearrange the outer bound of the sum rate as \eqref{EQ_genie_UB_sum_rate_desire}. This seems to be a strict condition but can be relaxed as long as the channel distributions follow proper stochastic orders, which will be explained in the latter part of this proof.} $|\tilde{h}_{12}|\geq |\tilde{h}_{22}|$ by definition of US IC as
\begin{align}
&\mathds{E}[h( \tilde{H}_{22}^n X_2^n+N_2^n |\tilde{\bm H}=\tilde{\bm h}) -h( \tilde{H}_{12}^n X_2^n+N_1^n |\tilde{\bm H}=\tilde{\bm h})]\notag\\
\overset{(a)}=& \mathds{E}[h( X_2^n+\tilde{N}_2^n|\tilde{\bm H}=\tilde{\bm h}) -h( X_2^n+\tilde{N}_1^n|\tilde{\bm H}=\tilde{\bm h}) +2\log(|\tilde{H}_{22}^n|/|\tilde{H}_{12}^n|)]\notag\\
\overset{(b)}=& \mathds{E}[h( X_2^n+\tilde{N}_1^n+N^n|\tilde{\bm H}=\tilde{\bm h}) -h( X_2^n+\tilde{N}_1^n|\tilde{\bm H}=\tilde{\bm h}) +2\log(|\tilde{H}_{22}^n|/|\tilde{H}_{12}^n|)]\notag\\
\overset{(c)}=& \mathds{E}[h( N^n+X_2^n+\tilde{N}_1^n|\tilde{\bm H}=\tilde{\bm h}) \!-\!h( N^n+X_2^n+\tilde{N}_1^n|N^n,\tilde{\bm H}=\tilde{\bm h}) \!+2\log(|\tilde{H}_{22}^n|/|\tilde{H}_{12}^n|)]\notag\\
\overset{(d)}=& \mathds{E}[I( N^n;N^n+X_2^n+\tilde{N}_1^n|\tilde{\bm H}=\tilde{\bm h}) +2\log(|\tilde{H}_{22}^n|/|\tilde{H}_{12}^n|)]\notag\\
\overset{(e)}\leq& \mathds{E}[I( N^n;N^n+\tilde{N}_1^n|\tilde{\bm H}=\tilde{\bm h}) +2\log(|\tilde{H}_{22}^n|/|\tilde{H}_{12}^n|)]\notag\\
\overset{(f)}\leq& \sum_{k=1}^n\mathds{E}\left[h(N_{2,k}|\tilde{\bm H}=\tilde{\bm h})-h(N_{1,k}|\tilde{\bm H}=\tilde{\bm h})\right],\label{EQ_genie_UB_sum_rate_2}
\end{align}
where in (a), $\tilde{N}_{2,k}\sim \mathcal{CN}(0,|\tilde{H}_{22,k}|^{-2})$ and $\tilde{N}_{1,k}\sim \mathcal{CN}(0,|\tilde{H}_{12,k}|^{-2})$; in (b) we define ${N}_k\sim \mathcal{CN}(0,|\tilde{H}_{22,k}|^{-2}-|\tilde{H}_{12,k}|^{-2})$, which uses the assumption $|\tilde{h}_{12}|\geq |\tilde{h}_{22}|$, while $N^n=[N_1,\,N_2,\,\cdots N_n]^T$ and $N^n$ is independent of $\tilde{N}^n$ and $X_2^n$ in the first term; in (c), we use the same assumption in (b) in the second term; in (d), we treat $N^n$ as the transmitted signal and $X_2^n+\tilde{N}_1^n$ as an equivalent noise at the receiver; in (e) we apply data processing inequality with the Markov chain: $N^n- \tilde{Y}_2^n- \tilde{Y}_1^n$, where $\tilde{Y}_1^n\triangleq X_2^n+\tilde{N}_1^n+N^n$ and $\tilde{Y}_2^n\triangleq \tilde{N}_1^n+N^n$; in (f) we use the fact that $N_1^n$ and $N_2^n$ are i.i.d., respectively and the assumption that $|\tilde{h}_{12}|\geq |\tilde{h}_{22}|$.


After substituting \eqref{EQ_genie_UB_sum_rate_2} into \eqref{EQ_genie_UB_sum_rate_1}, we obtain
\begin{align}\label{EQ_genie_UB_sum_rate_desire}
R_1+R_2&\leq \frac{1}{n}\sum_{k=1}^n I(X_{1k},X_{2k};Y_{1k}|\tilde{\bm H})\triangleq I(X_1,X_2;Y_1|\tilde{\bm H},Q),
\end{align}
where $Q$ is an auxiliary time sharing random variable uniformly distributed over $\{1,\cdots,\, n\}$. To proceed, we apply the result in \cite{Shamai_MAC_no_CSIT}, wherein the capacity region of a MAC is derived for cases in which only the receiver has perfect CSI but the transmitter has only statistical CSI. For such channels, the optimal input distribution is proved to be Gaussian. Note that since the transmitter has no instantaneous CSIT, we can neither apply power nor rate adaptation over time. As a result, maximum powers $P_1$ and $P_2$ are always used. Then we obtain \eqref{EQ_genie_UB_sum_rate_final1}.
Likewise, when the genie provides the interference only to the first receiver, we can get
\begin{align}\label{EQ_genie_UB_sum_rate_final2}
R_1+R_2\leq \mathds{E}\left[C\left(H_{21}P_1+H_{22}P_2\right)\right].
\end{align}
After comparing the outer bounds \eqref{EQ_genie_UB_sum_rate_final1} and \eqref{EQ_genie_UB_sum_rate_final2} to \eqref{EQ_CMAC} and \eqref{EQ_CMAC_def}, we can observe that decoding both messages at each receiver can achieve the capacity region outer bound.

To guarantee that the new channels constructing by \eqref{EQ_strong_IC_constraint} and \eqref{EQ_strong_IC_mapping} are equivalent to the original one, we need to verify the same marginal property:
\begin{align}
f_{Y_1|X_1X_2}=f_{Y_1'|X_1X_2} \mbox{ and } f_{Y_2|X_1X_2}=f_{Y_2'|X_1X_2},\label{EQ_SMP_IC}
\end{align}
where the received signals in the equivalent channel after the coupling are:
\begin{align}
Y_1'= H_{11}'X_1+H_{12}'X_2+N_1,\,\,Y_2'= H_{21}'X_1+H_{22}'X_2+N_2,
\end{align}
$\{H_{jk}'\},\,j,\,k=1,\,2$ follow \eqref{EQ_strong_IC_mapping}. For GICs where the noises are independent to channel gains, it suffices to prove:
\begin{align}\label{EQ_SMP_joint_distr_channel_gain}
f_{H_{22}H_{21}}=f_{H_{22}'H_{21}'}\mbox{ and }f_{H_{12}H_{11}}=f_{H_{12}'H_{11}'}.
\end{align}
The first term in \eqref{EQ_SMP_joint_distr_channel_gain} can be proved by:
\begin{align}\label{EQ_SMP_joint_distr_channel_gain2}
f_{H_{22}H_{21}}\overset{(a)}=f_{H_{22}}f_{H_{21}}\overset{(b)}=f_{H_{22}'}f_{H_{21}'}\overset{(c)}=f_{H_{22}'H_{21}'},
\end{align}
where (a) is from the assumption of the mutual independence between channel gains; (b) is from the existence of $H_{22}'=_d H_{22}$ and $H_{21}'=_d H_{21}$ due to the coupling; (c) is due to the selection $H_{21}' = F_{H_{21}}^{-1}(U_1)$, $H_{22}' = F_{H_{22}}^{-1}(U_2)$, and $U_1$ is independent of $U_2$, which leads to the fact that $H_{21}'$ is independent of $H_{22}'$. The same steps are valid for the second term in \eqref{EQ_SMP_joint_distr_channel_gain}, which completes the proof. \QEDA


\renewcommand{\thesection}{Appendix IV}
\section{Proof of Theorem \ref{Th_cond_very_strong_IC}}\label{APP_Th_cond_very_strong_IC}
We divide the proof into three parts: the feasibilities of $\mathcal{S}_1$ and $\mathcal{S}_2$ and the optimality of Gaussian input. Recall that the definition of a GIC with instantaneous very strong interference is
\begin{align}\label{EQ_AWGN_very_strong_IC_cond}
z_1\geq h_{11} \mbox{ and } \, z_2\geq h_{22},
\end{align}
where $z_1$ and $z_2$ are the realizations of $Z_1$ and $Z_2$, respectively, as defined in Theorem \ref{Th_cond_very_strong_IC}. Similar to the proof steps in \ref{APP_Th_cond_strong_IC}, we can reformulate the constraint for the case with statistical CSIT from \eqref{EQ_AWGN_very_strong_IC_cond} as:
\begin{align}\label{EQ_AWGN_very_strong_IC_cond_RV1}
Z_1\geq_{st} H_{11}\mbox{ and } Z_2\geq_{st} H_{22},
\end{align}
from which we have the coupling: $Z_1'=_d Z_1$, $H_{11}'=_d H_{11}$, $Z_2'=_d Z_2$, and $H_{22}'=_d H_{22}$. The remaining task is to prove that the conditions in the two sets $\mathcal{S}_1$ and $\mathcal{S}_2$ suffice to validate the same marginal property. We prove $f_{H_{22}H_{21}}=f_{H_{22}'H_{21}'}$ as follows.

We first prove the feasibility of $(\mathcal{A}_1,\,f'_1 )$. Note that since $H_{21}'$ and $H_{22}'$ are dependent due to the coupling in \eqref{EQ_VS_f1_prime_2nd} and also $H_{21}'\neq_d H_{21}$, the steps performed in the RHS of (b) and (c) in \eqref{EQ_SMP_joint_distr_channel_gain2} cannot be performed here. Hence we aim to prove:
\begin{align}\label{EQ_SMP_VS_GIC}
f_{H_{22}'H_{21}'}=f_{H_{22}H_{21}}\overset{(a)}=f_{H_{22}}f_{H_{21}}\overset{(b)}=f_{H_{22}'}f_{H_{21}},
\end{align}
where (a) is from the assumption of mutually independent channel gains and (b) is due to the coupling. From \eqref{EQ_SMP_VS_GIC} by Bayes' rule, it is equivalent to prove that
\begin{align}
f_{H_{21}'|H_{22}'}=f_{H_{21}}
\end{align}
is satisfied. Fix an arbitrarily constant $h\in\mbox{supp}\{H_{22}'\}$ in the following steps, we have:
\begin{align}
f_{H_{21}'|H_{22}'=h}&\overset{(a)}=f_{Z_1'(1+P_2h)|H_{22}'=h}\notag\\
&\overset{(b)}=f_{Z_1(1+P_2h)|H_{22}=h}\notag\\
&\overset{(c)}=f_{H_{21}|H_{22}=h},\notag
\end{align}
where (a) is from the definition of $Z_1'$; (b) results from the fact that $Z_1'=_d Z_1$ while both $Z_1$ and $Z_1'$ are the same function of $h$ and also $\mbox{supp}\{H_{22}'\}=\mbox{supp}\{H_{22}\}$ due to $H_{22}'=_d H_{22}$; (c) is from the definition of $Z_1$. Accordingly, we obtain
\begin{align}
f_{H_{21}'|H_{22}'}=f_{H_{21}|H_{22}}\overset{(a)}=f_{H_{21}},
\end{align}
where (a) is due to the independence between $H_{21}$ and $H_{22}$. The same steps with a different condition, i.e., the independence between $H_{12}$ and $H_{11}$, are valid for $f_{H_{11}H_{12}}=f_{H_{11}'H_{12}'}$, which completes the proof of $(\mathcal{A}_1,\,f'_1 )$.

To prove the feasibility of $(\mathcal{A}_2,\,f'_2 )$, we consider the case in which channel gains can be correlated. Again we only prove $f_{H_{22}H_{21}}=f_{H_{22}'H_{21}'}$ as above. We first define a mapping of random variables: $(H_{21},\,H_{22})\mapsto (Z_1,\,W_2)$, where $W_2\triangleq H_{22}$ is a trivial mapping. It is clear that the mapping is bijective. It is also clear that the mapping $(H_{21}',\,H_{22}' )\mapsto (Z_1' ,\,W_2')$ is the same as $(H_{21},\,H_{22})\mapsto (Z_1,\,W_2)$. Hence, if
\begin{align}\label{EQ_1st_constraint_of_SMP_VeryStrongIC}
f_{Z_1W_2}=f_{Z'_1W'_2},
\end{align}
then
\begin{align}
f_{H_{21}H_{22}}=f_{H'_{21}H'_{22}},
\end{align}
since the two Jacobians are the same. As a result, the same marginal property to the second receiver holds, which can be directly extended to the first receiver. Now we further express the condition \eqref{EQ_1st_constraint_of_SMP_VeryStrongIC} in terms of the PDFs of $H_{21}$ and $H_{22}$. Recall that we can express $Z'_1$ and $W'_2$ from the first and second terms in \eqref{EQ_AWGN_very_strong_IC_cond_RV1} by the coupling as:
\begin{align}
Z'_1=F_{Z_1}^{-1}(U_1),\,\,W'_2=F_{H_{22}}^{-1}(U_2),
\end{align}
respectively, where $U_1\sim\mbox{ Unif}(0,1)$, $U_2\sim\mbox{ Unif}(0,1)$. Note that we do not specify the relation between $U_1$ and $U_2$ till this step. Then the joint CDF of $Z'_1$ and $W'_2$ can be derived as:
\begin{align}
F_{Z'_1,\,W'_2}(a,b) &=\mbox{Pr}(Z'_1\leq a,\,W'_2\leq b)\notag\\
&=\mbox{Pr}\left(F_{Z_1}^{-1}(U_1)\leq a,\,F_{H_{22}}^{-1}(U_2)\leq b\right)\notag\\
&=\mbox{Pr}\left(U_1\leq F_{Z_1}(a),\,U_2\leq F_{H_{22}}(b)\right)\notag\\
&\overset{(a)}=\mbox{Pr}\left(U\leq F_{Z_1}(a),\,U\leq F_{H_{22}}(b)\right)\notag\\
&=\mbox{Pr}\left(U\leq \min\left\{F_{Z_1}(a),\,F_{H_{22}}(b)\right\}\right)\notag\\
&=\min\left\{F_{Z_1}(a),\,F_{H_{22}}(b)\right\},\notag
\end{align}
where in (a) we select $U_1=U_2$.\footnote{This selection leads to a more stringent constraint but we can have explicit constraints in terms of the distributions of $H_{jk},\,j,k=1,2$.} Therefore, if $Z_1$ and $W_2$ have the joint CDF as \eqref{Constraint_SMP_very_strong}, then we have $F_{Z_1,W_2}(a,b)=F_{Z'_1,\,W'_2}(a,b)$. Similarly, if $Z_2$ and $W_1$ have the joint CDF as \eqref{Constraint_SMP_very_strong},
then we have $F_{Z_2,W_1}(a,b)=F_{Z'_2,\,W'_1}(a,b)$, which validates the same marginal property.

Now we derive the ergodic capacity region of GIC with uniformly very strong interference, i.e., at each time instant it is a very strong IC, under statistical CSIT and then the condition of the uniformly very strong interference can be generalized as $\mathcal{A}_1$ or $\mathcal{A}_2$. Since there is no sum rate constraint in the capacity region of the IC with very strong interference, we can solve the optimal input distributions of GIC with statistical CSIT by considering
\begin{align}
&\underset{\substack{f_{X_1},\,f_{X_2}:\\
\mathds{E}[|X_1|^2]\leq P_1,\,\mathds{E}[|X_2|^2]\leq P_2}}{\arg\max } I( X_1;Y_1|X_2,H_{11},H_{12} ) + \mu\cdot I( X_2;Y_2|X_1,H_{22},H_{21} )\notag\\
=&\underset{\substack{f_{X_1},\,f_{X_2}:\\
\mathds{E}[|X_1|^2]\leq P_1,\,\mathds{E}[|X_2|^2]\leq P_2}}{\arg\max } I( X_1;H_{11}X_1 + Z_1|H_{11} ) + \mu\cdot I( X_2;H_{22}X_2 + Z_2|H_{22} ),\label{EQ_very_strong_equivalent2}
\end{align}
where $\mu\in\mathds{R}^+$. It is clear that for each $\mu$, \eqref{EQ_very_strong_equivalent2} can be maximized by Gaussian inputs, i.e., $X_1\sim\mathcal{CN}(0,P_1)$ and $X_2\sim\mathcal{CN}(0,P_2)$. Then the capacity region can be described by \eqref{EQ_CMAC2}.


\QEDA

\renewcommand{\thesection}{Appendix V}
\section{Proof of Theorem \ref{Th_WTC}}\label{APP_pf_WTC}

The achievability of \eqref{EQ_CWTC} can be derived by substituting $U=X\sim N(0,P_T)$ into the secrecy capacity in \cite[(8)]{csiszar1978broadcast}. In the following, we focus on the derivation of the outer bound of the secrecy capacity. In particular, we adopt the coupling scheme to show that Gaussian input is optimal for the outer bound and also show that the outer bound matches the inner bound. In the following, we first verify the validity of using the coupling scheme under the CSI assumption at Bob and Eve.

We require the original and the equivalent WTCs to have: 1) the same error probability $\mbox{Pr(}W\neq\hat{W})$ and,  2) the same equivocation rate $\frac{1}{n}h(W|Z^n)$, where $W$ and $\hat{W}$ are the transmitted and the detected secure messages at Alice and Bob, respectively. The first requirement is valid because the coupling scheme does not change Bob's channel distribution. Checking the second requirement is more involved. The reason is that we have asymmetric knowledge of the CSI at Bob and Eve. In general, to design for the worst case we assume that Eve has more knowledge of the CSI than Bob. As mentioned previously, a common assumption is that Bob knows perfectly the realization of his own channel $H$ but Eve knows perfectly both $H$ and $G$. The corresponding equivocation rate is described by
\begin{align}\label{EQ_original_equivocation_rate}
\frac{1}{n}h(W|Z^n, H^n, G^n),
\end{align}
and the calculation is determined by $f(w)f(x^n|w)f(z^n, h^n, g^n|x^n)$. If we directly apply the coupling scheme to $H$ and $G$, we will have the equivocation rate as $\frac{1}{n}h(W|\tilde{Z}^n, \tilde{H}^n, \tilde{G}^n)$, whose calculation relies on $f(\tilde{z}^n, \tilde{h}^n, \tilde{g}^n|x^n)$. Note that $f(z^n, h^n, g^n|x^n)$ may not be identical to $f(\tilde{z}^n, \tilde{h}^n, \tilde{g}^n|x^n)$ because coupling only guarantees the same marginal property but not same joint distribution. More specifically, the correlation between $H^n$ and $G^n$ can be arbitrary. In contrast, from coupling the correlation between $\tilde{H}^n $ and $\tilde{G}^n$ cannot be arbitrary, i.e., it is fixed by the marginal CDFs $F_{H}$ and $F_{G}$ and also the uniformly distributed $U$ when \eqref{EQ_WTC_f} is exploited. To avoid this inconsistence, we consider a new wiretap channel where Eve only knows $g^n$ with equivocation rate as $h(W|Z^n, G^n)$ calculated according to $f(z^n, g^n|x^n)$. Note that the secrecy capacity of the new GWTC is no less than the original one, since Eve here knows less CSI than in the original setting. Therefore, we derive the secrecy capacity of this new WTC as an outer bound of the original WTC. After applying the coupling scheme to the new WTC, we have the equivalent Eve's channel $\tilde{G}$ and the equivocation rate becomes $\frac{1}{n}h(W|\tilde{Z}^n, \tilde{G}^n)$, whose calculation is according to $f(\tilde{z}^n, \tilde{g}^n|x^n)$. A sufficient condition to ensure $\frac{1}{n}h(W|\tilde{Z}^n, \tilde{G}^n)=\frac{1}{n}h(W|{Z}^n,{G}^n)$ is that $f(z, g|x)=f(\tilde{z}, \tilde{g}|x)$, which can be attained by the coupling operation and is verified as follows:
\begin{align}
f(z,g|x)&=f(z|x,g)f(g|x)\notag\\
&\overset{(a)}=f(z|x,g)f(g)\notag\\
&\overset{(b)}=f(z|x,g)f(\tilde{g})\notag\\
&\overset{(c)}=f(\tilde{z}|x,\tilde{g})f(\tilde{g})\notag\\
&\overset{(d)}=f(\tilde{z},\tilde{g}|x),\notag
\end{align}
where (a) is due to the assumption of statistical CSIT, then $X$ and $G$ are independent; in (b), $f(g)=f(\tilde{g})$ is due to the same marginal property of the coupling operation; (c) comes from the fact that in the equivalent channel $\tilde{z}=\sqrt{\tilde{g}}x+n_Z$, the noise distribution is the same as that in the original channel and $N_Z$, $X$, and $\tilde{G}$ are mutually independent; in (d) we follow the steps in the first two equalities, reversely. Then we know that the new WTC where Eve only knows $g^n$ is equivalent to that after being applied coupling.

Based on the above discussion, we can construct a WTC equivalent to the new WTC, where Eve only knows $G$. Since $H\geq_{st} G$, the equivalent WTC is degraded, whose secrecy capacity is known as
\begin{align}
C_S'&\triangleq\max_{f_X}\,\,I(X;\tilde{Y}|\tilde{H})-I(X;\tilde{Z}|\tilde{G})\notag\\
&\overset{(a)}=\max_{f_X}\,\,I(X;\tilde{Y}|\tilde{H})-I(X;\tilde{Z}|F_G^{-1}(F_H(\tilde{H})))\notag\\
&\overset{(b)}=\max_{f_X}\,\,I(X;\tilde{Y}|\tilde{H})-I(X;\tilde{Z}|\tilde{H})\notag\\
&\overset{(c)}=\max_{f_X}\,\,I(X;\tilde{Y},\tilde{Z}|\tilde{H}),\label{EQ_Cs_new_WTC}
\end{align}
where (a) uses the relation $\tilde{G}=F_{G}^{-1}(U)$ and $\tilde{H}=F_{H}^{-1}(U)$; (b) uses the fact that $F_G^{-1}(F_H(\cdot))$ is a bijective mapping due to the generalized inverse; (c) uses the degradedness: $X-Y-Z$.
In addition, because
\begin{align}
\arg\max_{f_X}\, I(X;\tilde{Y}|\tilde{Z},\tilde{H}) =\arg\max_{f_X}\, h(\tilde{Y}|\tilde{Z},\tilde{H}),
\end{align}
we can extend \cite[Lemma 2]{Khisti_MISOME} to show that
\begin{align}
h(\tilde{Y}|\tilde{Z},\tilde{H})&\overset{(a)}=h(\tilde{Y}-\alpha \tilde{Z}|\tilde{Z},\tilde{H})\notag\\
&\overset{(b)}\leq\mathds{E}_{\tilde{H}}[h(\tilde{Y}-\alpha \tilde{Z}|\tilde{H}=\tilde{h})],\notag
\end{align}
where in (a) $\alpha$ is the linear minimum mean square error estimator of $\tilde{Y}$
from $\tilde{Z}$; in (b) the inequality is due to the fact
that conditioning only reduces differential entropy while the
equality holds by Gaussian $X$. Then both $\tilde{Y}$ and $\tilde{Z}$ are Gaussian if $X$ is Gaussian. After substituting Gaussian input into the definition of $C_S'$ in \eqref{EQ_Cs_new_WTC} with full power usage since power allocation can not be done due to statistical CSIT, we can get \eqref{EQ_CWTC}, which completes the proof. \QEDA

\renewcommand{\thesection}{Appendix VI}
\section{Proof of Theorem \ref{Th: st_stochastic_process}}\label{APP_pf_memory}
To be self-contained, we restate the \textit{strong stochastic order} \cite[Theorem 6.B.31]{shaked_stochastic_order}, which is a sufficient condition for the usual stochastic order among two random vectors or random processes.

\begin{Theorem}\cite[Theorem 6.B.31]{shaked_stochastic_order}\normalfont\label{Th_usual_st_random_process}
Let $\{X(0),\,X(1),\,X(2),\cdots\}$ and $\{Y(0),\,Y(1),\,Y(2),\cdots\}$ be two discrete-time random processes. If
\begin{align}
X(0)\leq_{st} Y(0),
\end{align}
and if
\begin{align}\label{EQ_usual_st_random_process_2st_cond}
[X(i)\,|\,X(i-1)=x(i-1),\cdots,X(1)=x(1)]\leq_{st}[Y(i)\,|\,Y(i-1)=y(i-1),\cdots,Y(1)=y(1)],
\end{align}
whenever
\begin{align}\label{EQ_usual_st_random_process_1nd_cond}
x(j)\leq y(j),\,j=1,\,2,\cdots,i-1,\,\,\,i=1,\,2,\cdots
\end{align}
then $\{X(m)\}\leq_{st}\{Y(m)\}$, $m\in\mathds{N}^0$.
\end{Theorem}

Note that $\{X(m)\}\leq_{st}\{Y(m)\}$ implies $X(m)\leq_{st}Y(m),\,\forall m$. From coupling we know that there exist $\{\tilde{H}_2(m)\}$ and $\{\tilde{H}_1(m)\}$ such that $\{\tilde{H}_1(m)\}=_{st}\{H_1(m)\}$, $\{\tilde{H}_2(m)\}=_{st}\{H_2(m)\}$, and $\mathrm{Pr}(\tilde{H}_1(m)\leq \tilde{H}_2(m))=1$, $\forall\,m\in\mathds{N}^0$, if $\{H_1(n)\}\leq_{st}\{H_2(n)\}$. Assume both channels $\{H_1(m)\}$ and $\{H_2(m)\}$ with a $k$-th order Markov structure. We can therefore modify the conditions in \eqref{EQ_usual_st_random_process_2st_cond} and \eqref{EQ_usual_st_random_process_1nd_cond} as:
\begin{align}\label{EQ_Kth_Markov_chain}
&[H_1(m)|H_1(m-1)=h_1(m-1),\cdots,H_1(m-k)=h_1(m-k)]\notag\\
\leq_{st}&[H_2(m)|H_2(m-1)=h_2(m-1),\cdots,H_2(m-k)=h_2(m-k)],\,\,\forall m\in\mathds{N}
\end{align}
with $ {h_1}(j)\leq  {h_2}(j)$ for all $j<m$. Note that the relation between $Y_i$ and $Z_i$ of an equivalently degraded channel in \eqref{EQ_stochastic_degraded} is described by $\prod_{i=1}^n \tilde{p}(z_i|y_i)$, i.e., comparing $z_i$ and $y_i$ element-wisely is sufficient. Therefore, \eqref{EQ_Kth_Markov_chain} implies  \eqref{EQ_stochastic_degraded}.

Based on the given transition matrices $\bm P$ and $\bm Q$, we can further simplify the constraint \eqref{EQ_Kth_Markov_chain} for the case $m> k$. Recall that the $j$-th entry of $\bm {p_{i}}$ and $\bm{q_{i}}$ are the transition probabilities from the $i$-th super state to the $j$-th super state of the Markov processes $\{H_1(m)\}$ and $\{H_2(m)\}$, respectively. Given $\bm {p_{i}}$ and $\bm {q_{i}}$, $i\in\{1,2,\cdots, N^k\}$, we can form the corresponding CCDF matrices, respectively, as $\bar{F}_{\bm P}=[\bar{F}_{{\bm {p}}_1}^T,\,\bar{F}_{{\bm {p}}_2}^T,\,\cdots,\bar{F}_{\bm {p}_{N^k}}^T]^T$ and $\bar{F}_{\bm Q}=[\bar{F}_{\bm {q}_1}^T,\,\bar{F}_{\bm {q}_2}^T,\,\cdots,\bar{F}_{\bm {q}_{N^k}}^T]^T$, where $\bar{F}_{\bm {p}_i}$ and $\bar{F}_{\bm {q}_i}$ are the CCDF vectors derived by $\bm {p}_i$ and $\bm {q}_i$, respectively. From Definition \ref{shaked_stochastic_order}, for $m>k$ we can equivalently express \eqref{EQ_Kth_Markov_chain} by $\bar{F}_{\bm {p}_l}(n)\leq \bar{F}_{\bm {q}_s}(n),\,\forall\,n$, with the constraint ${h_1}(j)\leq  {h_2}(j),\,\mbox{ for the time index }j<m$. To fulfill the constraints ${h_1}(j)\leq {h_2}(j)$, $j<m$, we choose the row indices $l$ and $s$  of the transition matrices of $\bm P$ and $\bm Q$, respectively, such that $g(l)\leq g(s)$ is ensured, which is due to the definition of the mapping $g$ in \eqref{EQ_g} and also the state values are listed in an increasing order. Then we use these $\{(l,s)\}$ to select feasible current channel states $H_2(m)$ and $H_1(m)$ by comparing the CCDF vectors in $\bar{F}_{\bm P}$ and $\bar{F}_{\bm Q}$. By this way, we attain \eqref{EQ_Markov_condition2} and \eqref{EQ_Markov_condition3}. Combining with \eqref{EQ_Markov_condition1}, we obtain the sufficient conditions to attain $\{H_1(m) \}\leq_{st}\{H_2(m) \}$, which implies the degradedness and completes the proof. \QEDA


\bibliographystyle{IEEEtran}
\renewcommand{\baselinestretch}{2}
\bibliography{IEEEabrv,../SecrecyPs2,2}
\end{document}